\documentclass[12pt]{article}
\usepackage{amsmath}
\usepackage{amsthm}
\usepackage{amssymb}
\usepackage{amsfonts}
\usepackage{epic}
\usepackage{eepic}
 \usepackage{times}
\usepackage[matrix,arrow]{xy}
\usepackage{macros}
\usepackage{epsfig}

\theoremstyle{plain}

\newtheorem{thm}{Theorem}
\newtheorem{propn}{Proposition}
\newtheorem{lem}{Lemma}

\theoremstyle{remark}
\newtheorem{rem}{Remark}

\newcommand{\sst}{\scriptscriptstyle}

\renewcommand{\1}{\one}
\renewcommand{\2}{\two}

\newcommand{\beq}{\begin{equation}}
\newcommand{\eeq}{\end{equation}}

\newcommand{\pa}{\partial}
\newcommand{\ot}{\otimes}
\newcommand{\ra}{\to}

\newcommand{\SRN}{{\rm N}}

\newcommand{\fsl}{{\mathfrak s}{\mathfrak l}}

\newcommand{\al}{\alpha}

\newcommand{\be}{\beta}
\newcommand{\ga}{\gamma}

\newcommand{\de}{\delta}
\newcommand{\De}{\Delta}
\newcommand{\ep}{\epsilon}
\newcommand{\la}{\lambda}

\newcommand{\bz}{\bar{z}}

\newcommand{\CA}{{\mathcal A}}
\newcommand{\CB}{{\mathcal B}}

\newcommand{\CC}{{\mathcal C}}
\newcommand{\CD}{{\mathcal D}}

\newcommand{\CH}{{\mathcal H}}

\newcommand{\CL}{{\mathcal L}}
\newcommand{\CM}{{\mathcal M}}

\newcommand{\CO}{{\mathcal O}}

\newcommand{\CQ}{{\mathcal Q}}
\newcommand{\CR}{{\mathcal R}}

\newcommand{\CU}{{\mathcal U}}

\newcommand{\SA}{{\mathsf A}}
\newcommand{\SB}{{\mathsf B}}
\newcommand{\SC}{{\mathsf C}}
\newcommand{\SD}{{\mathsf D}}

\newcommand{\SH}{{\mathsf H}}

\newcommand{\SM}{{\mathsf M}}

\newcommand{\SO}{{\mathsf O}}

\newcommand{\SQ}{{\mathsf Q}}

\newcommand{\ST}{{\mathsf T}}
\newcommand{\SU}{{\mathsf U}}

\newcommand{\SW}{{\mathsf W}}

\newcommand{\SY}{{\mathsf Y}}

\renewcommand{\sf}{{\mathsf f}}

\newcommand{\spp}{{\mathsf p}}

\newcommand{\su}{{\mathsf u}}
\newcommand{\sv}{{\mathsf v}}

\newcommand{\0}{{\mathfrak 0}}
\newcommand{\one}{{\mathfrak 1}}
\newcommand{\two}{{\mathfrak 2}}

\newcommand{\BB}{{\mathbb B}}
\newcommand{\BR}{{\mathbb R}}

\newcommand{\BC}{{\mathbb C}}

\newcommand{\BS}{{\mathbb S}}
\newcommand{\BT}{{\mathbf t}}

\newcommand{\BZ}{{\mathbb Z}}

\newcommand{\srn}{{\sst\rm N}}
\newcommand{\SRM}{{\rm M}}
\newcommand{\srm}{{\sst\rm M}}

\newcommand{\rf}[1]{(\ref{#1})}
\renewcommand{\bz}{{\mathbf z}}
\newcommand{\en}{{\rm e}_\SRN}

\newcommand{\aufz}
{\begin{list}{$\bullet$}{\topsep0cm \itemsep0cm \parsep0cm}}
\newcommand{\eaufz}{\end{list}}

\begin{document}


\title{The Sine-Gordon model revisited I}

\author{G. Niccoli$^{(1)}$, J. Teschner$^{(1)}$}

\address{$^{(1)}$ Notkestr. 85, 22603 Hamburg, Germany}

\maketitle

{\vspace{-7cm} \tt {DESY 09-170}}

\vspace{8cm}

{\bf Abstract}
\\
\begin{quotation}
\hspace{-0.7cm} We study integrable lattice regularizations of the Sine-Gordon model with the help of the Separation of Variables method of Sklyanin and the Baxter $\SQ$-operators. This leads us to the complete characterization of the spectrum (eigenvalues and eigenstates), in terms of the solutions to 
the Bethe ansatz equations. The completeness of the set of states that
can be constructed from the solutions to the Bethe ansatz equations
is proven by our approach.
\end{quotation}

\newpage

\tableofcontents

\newpage

\section{Introduction}

\subsection{Motivation}

The study of the Sine-Gordon model has a long history. It has in particular
served
as an important
toy model for interacting quantum field theories. The integrability
of this model gives access to detailed non-perturbative information
about various characteristic quantities, which allows one to
check physical ideas about quantum
field theory against exact quantitative results.

It is particularly fascinating to compare the Sine-Gordon model
with the Sinh-Gordon model. The Hamiltonian density
$h_{SG}$ of the Sine-Gordon model and the
corresponding object $h_{ShG}$ of the Sinh-Gordon model,
\begin{equation}\label{Hdef}
H\,=\,\int_0^{R}\frac{dx}{4\pi}\;h(x)\,,\qquad
\begin{aligned}
 h_{SG}^{} &\,=\,\Pi^2+(\pa_x\phi)^2+8\pi\mu\cos(2\be\phi)\,,\\
 h_{ShG}^{}& \,=\,\Pi^2+(\pa_x\phi)^2+8\pi\mu\cosh(2b\phi)\,,
\end{aligned}
\end{equation}
are related by analytic continuation w.r.t.
the parameter $\beta$ and setting $\beta=ib$.
The integrability of both
models is governed by the same algebraic structure
$\CU_q(\widehat{\fsl}_2)$ with $q=e^{-\pi i \be^2}$.
This leads one to expect that both models should be closely
related, or at least have the
same ``degree of complexity''.

The physics of these two models turns
out to be very different, though. Many of the key objects characteristic
for the respective quantum field theories are not related by analytic
continuation in the usual sense. While the Sine-Gordon
model has much richer spectrum of excitations and scattering
theory in the infrared (infinite $R$) limit, one may observe rather
intricate structures in the UV-limit of the Sinh-Gordon model \cite{Z2},
which turn out to be related to the Liouville theory \cite{ZZ,T,BT}.
These differences can be traced back to the fact that the periodicity
of the interaction term $8\pi\mu\cos(2\be\phi)$ of the Sine-Gordon model
allows one to treat the variable $\phi$ as angular variable parameterizing
a compact space, while $\phi$ is truly non-compact in the Sinh-Gordon model.

The qualitative differences between the Sine-Gordon and the Sinh-Gordon model
can be seen as a simple model for the differences between Nonlinear
Sigma-Models on compact and non-compact spaces respectively. This forms
part of our motivation to revisit the Sine-Gordon model in a way that
makes comparison with the Sinh-Gordon model easier.

\subsection{Open problems}

A lot of important exact results are known about the Sine-Gordon model.
Well-understood
are in particular the scattering theory in the
infinite volume. The spectrum of elementary particle excitations
and the S-matrix of the theory are known exactly  \cite{KT77,Za77,FST,Ko80}. Relatedly, there
is a wealth of information on the form-factors of local fields, see e.g. \cite{Sm92,BFKZ,LZ01} for
the state of the art and further references.
In the case of finite spacial volume, the nonlinear integral equations\footnote{This type of equations were before introduced in a different framework in \cite{KP91,KBP91}}
derived by Destri and De Vega \cite{DDV92,DDV94,DDV97,FMQR97,FRT98,FRT99} 
give a powerful tool for the study of
the finite-size corrections to the spectrum of the Sine-Gordon model.

However, there are several questions, some of them fairly basic,
where our understanding does not seem to be fully satisfactory.
We do not have exact results on correlation functions on the one hand,
or on expectation values of local fields in the finite volume
on the other hand at present.

Even the present level of understanding
of the spectrum of the model does not seem to be fully satisfactory.
The truth of the
commonly accepted hypothesis that the equations derived by
Destri and De Vega describe all of the states of the Sine-Gordon model
has not been demonstrated yet.
The approach of Destri and De Vega is
based on the Bethe ansatz in the fermionized version of
the Sine-Gordon model, the massive Thirring model \cite{DDV87}.
This approach a priori only allows one
to describe the states with even topological charge, and
it inherits from its roots in the algebraic Bethe ansatz some 
basic difficulties
like the issue of its completeness.

In the Bethe ansatz approach it is a long-standing
problem to prove that the set of states that is obtained in this
way is complete. Early attempts to show completeness used the
so-called string hypothesis which is hard to justify, and 
sometimes even incorrect. At the moment there are only a few
examples of integrable models where the completeness of the
Bethe ansatz has been proven, including the XXX Heisenberg model,
see \cite{MTV} and references therein. A similar result has not
been available for the Sine-Gordon model or its lattice
discretizations yet. One of the main results in this paper
is the completeness result for the lattice Sine-Gordon model. We prove a one-to-one correspondence
between eigenstates of the transfer matrix and the solutions to
a system of algebraic equations of the 
Bethe ansatz type. For brevity, we will refer to this result
as {\it completeness} of the Bethe ansatz. 
We furthermore show that the spectrum of the transfer matrix 
is simple in the case of odd number of lattice sites, and find the operator
which resolves the possible double degeneracy of the
spectrum of the transfer matrix in the case of even 
number of lattice sites.


\subsection{Our approach}

We will use a lattice regularization
of the Sine-Gordon model that is different from the
one used by Destri and De Vega. It goes back to \cite{FST,IK},
and it has more recently been studied in \cite{F94,FV94}.
For even number of lattice sites the model is related to the
Fateev-Zamolodchikov model \cite{FZ}, as was observed in \cite{FV94},
or more generally to the Chiral Potts model, as
discussed in the more recent works \cite{BBR,Ba08}.
This allows one to use some powerful algebraic tools developed 
for the study of the chiral Potts model \cite{BS} 
in the analysis of the lattice Sine-Gordon model. 

The issue of completeness of the Bethe ansatz had not 
been solved in any of these models yet. 
What allows us to address this issue is the 
combination of Separation of Variables method (SOV-method) of
Sklyanin \cite{Sk1,Sk2,Sk3}
with the use of the $\SQ$-operators introduced by Baxter \cite{Ba72}.
We will throughout be working with a certain number
of inhomogeneity parameters. It turns out that the SOV-method 
works in the case of generic inhomogeneity parameters where
the algebraic Bethe ansatz method fails.
It replaces the algebraic Bethe ansatz as a tool to 
construct the eigenstates of the transfer matrix
which correspond to the solutions of Bethe's equations.
In a future publication we will show that the results of
our approach are consistent with the
results of Destri and De Vega.

Another advantage of the lattice discretization used in 
this paper which 
may become useful in the future
is due to the fact that one directly works with the
discretized Sine-Gordon degrees of freedom, which is not the case
in the lattice formulation used by Destri and De Vega.
Working more directly with the Sine-Gordon degrees of freedom
should in particular 
be useful for the problem to calculate expectation
values of local fields. This in particular requires the determination of the SOV-representation of local fields analogously to what has been done in the framework of the algebraic Bethe ansatz in \cite{KMT99,MT00}. 
The SOV-method in principle offers a rather
direct way to the construction of the expectation values,
as illustrated in the case of the Sinh-Gordon model by the
work \cite{Lu}.

\vspace*{1mm}
{\par\small
{\em Acknowledgements.} We would like to
thank V. Bazhanov and F. Smirnov
for stimulating discussions, and J.-M. Maillet for interest in our work.

We gratefully acknowledge support
from the EC by the Marie Curie Excellence
Grant MEXT-CT-2006-042695.}

\section{Definition of the model}

\setcounter{equation}{0}

\subsection{Classical Sine-Gordon model}

The classical counterpart of the Sine-Gordon model is a dynamical system
whose degrees of freedom are described by
the field $\phi(x,t)$ defined for
$(x,t)\in[0,R]\times \BR$ with periodic boundary conditions $\phi(x+R,t)=\phi(x,t)$.
The dynamics of this model
may be described in the Hamiltonian form in terms of variables
$\phi(x,t)$, $\Pi(x,t)$, the Poisson brackets being
\[
\{\,\Pi(x,t)\,,\,\phi(x',t)\,\}\,\,=\,2\pi\,\de(x-x')\,.
\]
The time-evolution of an arbitrary observable $O(t)$ is then given as
\[
\pa_tO(t)\,=\,\{\,H\,,\,O(t)\,\}\,,
\]
with Hamiltonian $H$ being defined in \rf{Hdef}.

The equation of motion for the
Sine-Gordon model can be represented as a zero curvature condition,
\begin{equation}
[\,\pa_t-V(x,t;\lambda)\,,\,\pa_x-U(x,t;\lambda)\,]\,=\,0\,,
\end{equation}
with matrices $U(x,t;\lambda)$  and  $V(x,t;\lambda)$ being given by
\begin{equation}
\begin{aligned}\label{ZCC}
&U(x,t;\lambda)\,=\,\left(
\begin{matrix}i\frac{\be}{2}\Pi & -{i}m(\la e^{-i\be\phi}-\la^{-1}e^{i\be\phi})\\
-{i}m(\la e^{i\be\phi}-\la^{-1}e^{-i\be\phi}) & - i\frac{\be}{2}\Pi
\end{matrix}\right)\\
&V(x,t;\la)\,=\,\left(\begin{matrix} i\frac{\be}{2}\phi' & +{i}m(\la e^{-i\be\phi}+\la^{-1}e^{i\be\phi}) \\ +{i}m(\la e^{i\be\phi}+\la^{-1}e^{-i\be\phi})
& -i\frac{\be}{2}\phi'
\end{matrix}\right)
\end{aligned}
\end{equation}
and $m$ related to $\mu$ by $m^2=\pi \be^2\mu$.

\subsection{Discretization and canonical quantization}

In order to regularize the ultraviolet divergences that arise in the quantization of these models
we will pass to integrable lattice discretizations.
First discretize the field variables according to the standard recipe
\begin{equation*}
 \phi_n \equiv \phi(n\Delta) \,, \quad
 \Pi_n \equiv \De\Pi(n\Delta) \,,
\end{equation*}
where $\Delta=R/\SRN$ is the lattice spacing.
In the canonical quantization one would replace $\phi_n$, $\Pi_n$ by
corresponding quantum operators with commutation relations
\begin{equation}
[\,\phi_n\,,\,\Pi_n\,]\,=\,2\pi i\de_{n,m}\,.
\end{equation}
Planck's constant can be identified with $\be^2$ by means of a rescaling of the
fields.

The scheme of quantization of the Sine-Gordon model considered in this paper
will
deviate from the canonical quantization by using
$\su_n\equiv e^{i\frac{\be}{2}\Pi_n}$ and
$\sv_n\equiv e^{-i\beta\phi_n}$ as basic variables. For technical reasons we
will consider representations where both $\su_n$ and $\sv_n$ have discrete
spectrum. Let us therefore take a moment to explain why one may nevertheless
expect that the
resulting quantum theory will describe the quantum Sine-Gordon
model in the continuum limit.

First note (following the discussion in \cite{Za94})
that the periodicity of the potential $8\pi\mu \cos(2\beta\phi)$
in \rf{Hdef} implies that shifting the zero mode
$\phi_0\equiv \frac{1}{R}
\int_0^R dx \,\phi(x)$ by the amount $\pi/\beta$ is a symmetry.
In canonical quantization one could build
the unitary operator $\SW=e^{\frac{i}{2\be}R\spp_\0}$
which generates this symmetry
out of the zero mode $\spp_\0
\equiv\frac{1}{R}\int_0^R dx \,\Pi(x)$ of the conjugate
momentum $\Pi$.
$\SW$ should commute with the Hamiltonian $\SH$.
One may therefore diagonalize $\SW$ and $\SH$ simultaneously, leading
to a representation for the space of states in the form
\begin{equation}
\CH\,\simeq\,\int_{S_1}d\alpha\;\CH_\al\,\qquad{\rm where}\qquad
\SW\cdot\CH_\al\,=\,e^{i\al}\CH_\al\,.
\end{equation}
An alternative
way to take this symmetry into account in the construction of the
quantum theory is to construct the quantum theory separately for each
$\al$-sector. This implies that the field $\phi$ should be treated as
periodic with periodicity $\pi/\beta$, and that the conjugate
variables $\Pi_n$ have eigenvalues quantized in units
of $\beta$, with spectrum
contained in $\{\,2\al\be/ \SRN+4\pi \beta k\,;\,k\in\BZ\,\}$. The
spectrum of $\Pi_n$ is
such that the operator $\SW=e^{\frac{i}{2\be}R\spp_\0}$, with
$R\spp_\0$ approximated
by $\sum_{n=1}^\SRN \Pi_n$, is realized as the operator of
multiplication by $e^{i\al}$.

Let us furthermore note
that it is possible, and technically
useful to assume that the lattice field observable $\phi_n$ has discrete spectrum, which we will take
to be quantized in units of $\beta$. In order to see this,
note that the field $\phi(x)$ is not a well-defined observable due to
short-distance singularities, whereas smeared fields
like $\int_I dx\,\phi(x)$,
$I\subset [0,R]$ may be well-defined.
The observable $\int_I dx\,\phi(x)$
would in the lattice discretization be approximated by
\begin{equation}
\phi[I]\,\sim\,\sum_{n\Delta\in I} \Delta\phi_n\,.
\end{equation}
So even if $\phi_n$ is discretized in units of $\beta$, say, we find that the observable
$\phi[I]$ is quantized in units of $\Delta\beta$, which fills out a continuum for $\Delta\ra 0$.

\subsection{Non-canonical quantization}

As motivated above, we will use a quantization scheme based on the quantum counterparts of the  variables
$u_n$, $v_n$
$n=1,\dots,\SRN$ related to $\Pi_n$, $\phi_n$ as
\begin{equation}
u_n\,=\,e^{i\frac{\be}{2}\Pi_n}\,,\qquad
v_n\,=\,e^{-i\beta\phi_n}\,.
\end{equation}
The quantization of the
variables $u_n$, $v_n$ produces
operators $\su_n$, $\sv_m$ which satisfy the
relations
\begin{equation}\label{Weyl}
\su_n\sv_m=q^{\de_{nm}}\sv_m\su_n\,,\qquad{\rm where}\;\;
q=e^{-\pi i \beta^2}\,.
\end{equation}
We are looking for representations for the
commutation relations \rf{Weyl} which have discrete spectrum
both for $\su_n$ and $\sv_n$.
Such representations exist provided that
the parameter $q$ is a root of unity,
\begin{equation}
\beta^2\,=\,\frac{p'}{p}\,,\qquad p,p\in\BZ^{>0}\,.
\end{equation}
We will restrict our attention to the case $p$ odd and $p'$ even so that $q^{p}=1$.
It will often be convenient to parameterize $p$
as
\begin{equation}
p\;=\;2l+1\,,\qquad l\in\BZ^{\geq 0}\,.
\end{equation}
Let us consider the subset
$\BS_p=\{q^{2n};n=0,\dots,2l\}$ of the unit circle. Note that
$\BS_p=\{q^{n};n=0,\dots,2l\}$ since $q^{2l+2}=q$.
This allows us to
represent the operators $\su_n$, $\sv_n$ on the space of
complex-valued functions $\psi:\BS_p^{\SRN}\ra \BC$ as
\begin{equation}\label{reprdef}
\begin{aligned}
&\su_n\cdot\psi(z_1,\dots,z_\SRN)\,=\,u_n z_n\psi(z_1,\dots,
z_n,\dots,z_\SRN)\,,\\
&\sv_n\cdot\psi(z_1,\dots,z_\SRN)\,=\,v_n\psi(z_1,\dots, q^{-1} z_n,\dots,z_\SRN)\,.
\end{aligned}
\end{equation}
The representation is such that the operator $\su_n$
is represented as a multiplication operator.
The parameters
$u_n$, $v_n$ introduced in \rf{reprdef} can be interpreted as
``classical expectation values'' of the operators $\su_n$ and $\sv_n$.
The discussion in the previous subsection suggests that the $v_n$ will be irrelevant in
the continuum limit, while the average value of $u_n$ will be related to
the eigenvalue $e^{i\al}$ of $\SW$ via $u_n=\exp(i\beta^2{\al}/{\SRN})$.

\subsection{Lattice dynamics}\label{dyn}

There is a beautiful discrete time evolution that can be
defined in terms of the
variables introduced above which reproduces the Sine-Gordon
equation in the classical continuum limit \cite{FV94}.
It is simplest in the case where $u_n=1$, $v_n=1$,
$n=1,\dots,\SRN$. We will mostly\footnote{Except for Section \ref{SOV}.}
restrict to this case in the rest of this paper.

More general cases were treated in
\cite{BBR,Ba08}.

\subsubsection{Parameterization of the initial values}

As a convenient set of variables let us introduce the
observables $f_{k}$ defined as
\begin{equation} f_{2n}\,\equiv\,e^{-2i\beta\phi_n}\,,\qquad
f_{2n-1}\,\equiv\,
e^{i\frac{\beta}{2}(\Pi_n+\Pi_{n-1}-2\phi_n-2\phi_{n-1})}\,.
\end{equation}
These observables turn out to represent the initial data
for time evolution in a particularly convenient way.
The quantum operators $\sf_n$ which
correspond to the classical observables $f_n$ satisfy the
algebraic relations
\begin{equation}\label{funalg}
\sf_{2n\pm 1}\,\sf_{2n}\,=\,q^2\,\sf_{2n}\,\sf_{2n\pm 1}\,,\quad q=e^{-\pi i \beta^2}\,,
\qquad
\sf_{n}\,\sf_{n+m}\,=\,\sf_{n+m}\,\sf_{n}\;\;{\rm for}\;\;m\geq 2\,.
\end{equation}
There
exist simple representations of the algebra
\rf{funalg} which may  be constructed out of the
operators $\su_n$, $\sv_n$,
given by
\begin{equation}\label{f-uv}
\sf_{2n}\,=\,\sv_n^2\,,\qquad\sf_{2n-1}\,=\,\su_n^{}\su_{n-1}\,.
\end{equation}
The change of variables defined in \rf{f-uv} is invertible if $\SRN$ is odd.

\subsubsection{Discrete evolution law}

Let us now describe the discrete
time evolution proposed by
Faddeev and Volkov \cite{FV94}.
Space-time is replaced by the cylindric lattice
\[
\CL\,\equiv\,\big\{\,(\nu,\tau)\,,\,\nu\in\BZ/\SRN\BZ\,,\,\tau\in\BZ\,,\,\nu+\tau={\rm even}\,\big\}\,.
\]
The condition that $\nu+\tau$ is even means that the lattice is rhombic: The lattice points closest
to $(\nu,\tau)$ are $(\nu\pm 1,\tau+1)$ and $(\nu\pm 1,\tau-1)$.
We identify the variables $\sf_n$ with the initial values of a discrete "field" $\sf_{\nu,\tau}$ as
\[ \sf_{2r,0}\,\equiv\,\sf_{2r}\,,\qquad
\sf_{2r-1,1}\,\equiv\,\sf_{2r-1}\,.
\]
One may then extend the definition recursively to all $(\nu,\tau)\in\CL$ by means of the evolution law
\begin{equation}\label{Hirota}
{\sf}_{\nu,\tau+1}\,\equiv\,g_\kappa^{}\big(q\sf_{\nu-1,\tau}^{}\big)\cdot
\sf_{\nu,\tau-1}^{-{1}}\cdot g_\kappa^{}\big(q\sf_{\nu+1,\tau}^{}\big)
\,,
\end{equation}
with function $g$ defined as
\begin{equation}\label{gkappadef}
g_\kappa(z)\,=\,\frac{\kappa^2+z}{1+\kappa^2 z}
\end{equation}
where $\kappa$ plays the role of a scale-parameter of the theory.
We refer to \cite{FV94} for a nice discussion of the relation between the lattice evolution equation \rf{Hirota} and the classical Hirota
equation, explaining in particular how to recover
the Sine-Gordon equation in the classical continuum limit.

\subsubsection{Construction of the evolution operator}

In order to construct the  unitary operators $\SU$ that generate the time evolution \rf{Hirota}
let us introduce the function
\begin{align}\label{W-big}
& W_{\la}(q^{2n})\,=\,\prod_{r=1}^{n}
\frac{1+\la q^{2r-1}}{\la+ q^{2r-1}}\,,
\end{align}
which is cyclic, i.e. defined on $\BZ_p$. The function $W_\la(z)$ is a solution to the functional equation
\begin{align}\label{funrel}
& (z+\la)W_\la(qz)\,=\,(1+\la z)W_\la(q^{-1}z)\,,
\end{align}
which satisfies the unitarity relation
\begin{equation}
(W_\la(z))^*_{}\,=\,(W_{\la^*}^{}(z))^{-1}\,.
\end{equation}
Note in particular that $W_\la(z)$ is "even", i.e. $W_\la(z)=W_{\la}(1/z)$.
Further properties of this function are collected in Appendix A.

Let us then consider the operator $\SU$, defined as
\begin{equation}
\SU\,=\,\prod_{n=1}^{\SRN}W_{\kappa^{-2}}(\sf_{2n})
\cdot\SU_0\cdot\prod_{n=1}^{\SRN}W_{\kappa^{-2}}(\sf_{2n-1})\,,
\end{equation}
where $\SU_0$ is the parity operator that acts as $\SU_0^{}\cdot\sf_k^{}=\sf^{-1}_k\cdot\SU_0^{}$.
It easily follows from \rf{funrel} that $\SU$ is indeed the generator of the time-evolution \rf{Hirota},
\begin{equation}
\sf_{\nu,\tau+1}\,=\,\SU^{-1}\cdot \sf_{\nu,\tau-1}\cdot\SU\,.
\end{equation}
One of our tasks is to exhibit the integrability of this
discrete time evolution.

\section{Integrability}

\setcounter{equation}{0}

The integrability of the lattice Sine-Gordon model is known 
\cite{IK,FV,BKP93,BBR}. The most convenient way to formulate it
uses the Baxter $\SQ$-operators \cite{Ba72}.
These operators have been constructed 
for the closely related Chiral Potts model in \cite{BS}.
By means of the relation between the lattice Sine Gordon
model and the Fateev-Zamolodchikov model summarized in 
Appendix \ref{FZ} one may adapt these constructions 
to the formulation used in this paper. 
For the reader's convenience we will give a self-contained summary of the 
construction of the $\ST$- and $\SQ$-operators and of their 
relevant properties in the following section.

\subsection{$\ST$-operators}\label{T-op}

As usual in the quantum inverse scattering method, we will represent
the family $\CQ$ by means of a Laurent-polynomial $\ST(\la)$
which depends on the spectral parameter $\la$.
The definition of operators $\ST(\la)$
for the models in question is standard. It is of the general form
\begin{equation}\label{Mdef}
\ST^{}(\la)\,=\,{\rm tr}_{\BC^2}^{}\SM(\la)\,,\qquad \SM(\la)\,\equiv\,
L_\SRN^{}(\la/\xi_\SRN)\dots L_1^{}(\la/\xi_1)\,,
\end{equation}
where we have introduced inhomogeneity parameters $\xi_1,\dots,\xi_\SRN$ as a useful technical device.
The Lax-matrix may be chosen as
\begin{equation}\label{Lax}\begin{aligned}
 L^{\rm\sst SG}_n(\la)
  &= \frac{\kappa_n}{i} \left( \begin{array}{cc}i\,\su_n^{}(q^{-\frac{1}{2}}\kappa_n^{}\sv_n^{}+q^{+\frac{1}{2}}\kappa^{-1}_n\sv_n^{-1}) &
\la_n^{} \sv_n^{} - \la^{-1}_n \sv_n^{-1}  \\
 \la_n^{} \sv_n^{-1} - \la^{-1}_n \sv_n^{} &
i\,\su_n^{-1}(q^{+\frac{1}{2}}\kappa^{-1}_n\sv_n^{}+q^{-\frac{1}{2}}\kappa_n^{}\sv_n^{-1})
 \end{array} \right) .
\end{aligned}\end{equation}
An important motivation for the definitions \rf{Mdef}, \rf{Lax} comes from the fact that
the Lax-matrix $L^{\rm\sst SG}_n(\la)$ reproduces the Lax-connection $U(x)$ in the continuum limit.

The elements of the matrix $\SM(\la)$ will be denoted by
\begin{equation}\label{ABCD}
\SM(\la)=\left(\begin{matrix}\SA(\la) & \SB(\la)\\
\SC(\la) & \SD(\la)\end{matrix}\right)\,.
\end{equation}
They satisfy commutation relations that may be summarized
in the form
\begin{equation}\label{YBA}
R(\la/\mu)\,(\SM(\la)\ot 1)\,(1\ot\SM(\mu))\,=\,(1\ot\SM(\mu))\,(\SM(\la)\ot 1)R(\la/\mu)\,,
\end{equation}
where the auxiliary R--matrix is given~by
\begin{equation}\label{Rlsg}
 R(\la) =
 \left( \begin{array}{cccc}
 q^{}\la-q^{-1}\la^{-1} & & & \\ [-1mm]
 & \la-\la^{-1} & q-q^{-1} & \\ [-1mm]
 & q-q^{-1} & \la-\la^{-1} & \\ [-1mm]
 & & &  q\la-q^{-1}\la^{-1}
 \end{array} \right) \,.
\end{equation}
It will be useful for us to regard the definition \rf{Mdef} as the construction of operators
which generate a representation $\CR_\SRN$ of the so-called
Yang-Baxter
algebra defined by the quadratic relations \rf{YBA}. The representation
$\CR_\SRN$ is  characterized by the $4\SRN$ parameters $\kappa=(\kappa_1,\dots,\kappa_\SRN)$,
$\xi=(\xi_1,\dots,\xi_\SRN)$, $u=(u_1,\dots,u_\SRN)$ and $v=(v_1,\dots,v_\SRN)$.

The fact that the elements of $\SM(\la)$ satisfy the commutation relations \rf{YBA} forms the basis
for the application of the quantum inverse scattering method.
The mutual commutativity of the $\ST$-operators,
\begin{equation}
[\,\ST(\la)\,,\,\ST(\mu)\,]\,=\,0\,,
\end{equation}
follows from \rf{YBA} by standard arguments. The expansion of $\ST(\la)$ into powers of $\la$ produces
$\SRN$ algebraically independent operators $\ST_1,\dots,\ST_{\SRN}$.
Our main objective in the following will be the study of the spectral problem
for $\ST(\la)$. The importance of this spectral problem follows from the
fact that the time-evolution operator
$\SU$ of the lattice Sine-Gordon model will be shown to commute with $\ST(\la)$ in the next section.

\subsection{$\SQ$-operators}

Let us now introduce the Baxter $\SQ$-operators
$\SQ(\mu)$. These
operators are mutually commuting for
arbitrary values of the spectral parameters $\la$ and $\mu$, and satisfy a functional relation of
the form
\begin{equation}
\ST(\la)\SQ(\la)\,=\,{\tt a}(\la)\SQ(q^{-1}\la)+{\tt d}(\la)\SQ(q\la)\,,
\end{equation}
with $a(\la)$ and $d(\la)$ being certain model-dependent coefficient functions. The generator of lattice time
evolution will be constructed from the specialization of the $\SQ$-operators to certain values of the spectral
parameter $\la$, making the integrability of the evolution manifest.

\subsubsection{Construction}

In order to construct the $\SQ$-operators let us introduce the following renormalized version of the function $W_{\la}(z)$,
\begin{align}
& w_{\la}(q^{2n})
\,=\,
\prod_{r=1}^{n}\frac{1+\la q^{2r-1}}{\la+ q^{2r-1}}
\prod_{r=1}^{l}\frac{\la+ q^{2r-1}}{1+q^{2r-1}}\,,
\end{align}
The function $w_\la(z)$ is the unique  solution to the functional
equation \rf{funrel} which
is a  polynomial of order $l$ in $\la$ and which satisfies the
normalization condition $w_1(q^{2n})=1$.

The $\SQ$-operators can then be constructed in the form
\begin{equation}
\SQ(\la,\mu)\,=\,\SY(\la)\cdot(\SY(\mu^*))^{\dagger}\,,
\end{equation}
where $\SY(\la)$ is defined by its matrix elements $Y_\la^{}(\bz,\bz')\,\equiv\,\langle\,{\mathbf z}\,|\,\SY(\la)\,|\,{\mathbf z}'\,\rangle$
which read
\begin{equation}\label{Ydef}
Y_\la^{}(\bz,\bz')\,=\,\prod_{n=1}^{\SRN}\overline{w}_{\ep\la/\kappa_n\xi_n}^{}(z_n^{}/z_n')\,w_{\ep\la\kappa_n/\xi_n}^{}(z_n^{}z_{n+1}')\,,
\end{equation}
where $\ep=-iq^{-\frac{1}{2}}$, and $\overline{w}_\la(z)$ is the discrete Fourier transformation of $w(z)$,
\begin{align}\label{FT}
 \overline{w}_{\la}(z)\,=\,\frac{1}{p}\sum_{r=-l}^l
z^r\,w_{\la}(q^{r})\,,\qquad w_\la(y)\,=\,\sum_{r=-l}^{l} y^{-r}\,\overline{w}_{\la}(q^r)\,.
\end{align}
Note in particular the normalization condition $\overline{w}_1(q^r)=\de_{r,0}$.

Despite the fact that $\SQ(\la,\mu)$ is symmetric in $\la$ and  $\mu$, $\SQ(\la,\mu)=\SQ(\mu,\la)$ as follows from the
identity \rf{Yex} proven  in Appendix B, we will mostly consider $\mu$ as a fixed parameter which will later be chosen
conveniently. This being understood we will henceforth write $\SQ(\la)$ whenever the dependence of
$\SQ(\la,\mu)$ on $\mu$ is not of interest.

\subsubsection{Properties}

\begin{thm} \label{Qprop} $\quad$ --- $\quad${\bf Properties of $\ST$- and $\SQ$-operators}$\quad$ ---

{\sc (A) Analyticity}

The operator $\la^{\tilde{\SRN}}\ST(\la)$ is a polynomial in $\la^2$ of degree\footnote{Here, we use the notation ${\rm e}_\SRN=1$ for even $\SRN$, ${\rm e}_\SRN=0$ otherwise.} $\tilde{\SRN}:=\SRN+\en-1$ while the operator $\SQ(\la)$ is a polynomial in $\la$ of maximal degree $2l\SRN$.
In the case $\SRN$ odd the operators $\SQ_{2l\SRN}:=\lim_{\la\ra\infty}\la^{-2l\SRN}\SQ(\la)$ and $\SQ_0:=\SQ(0)$ are invertible operators and the normalization of the $\SQ$-operator can be fixed by $\SQ_{2l\SRN}={\rm id}$.

{\sc (B) Baxter equation}

The operators $\ST(\la)$ and $\SQ(\la)$ are related by the Baxter equation
\begin{equation}\label{BAX}
\ST(\la)\SQ(\la)\,=\,{\tt a}_\SRN^{}(\la)\SQ(q^{-1}\la)+{\tt d}_\SRN^{}(\la)\SQ(q\la)\,,
\end{equation}
with coefficient functions
\begin{equation}\label{addef}\begin{aligned}
& {\tt a}_\SRN^{}(\la)\,=\,(-i)^\SRN\prod_{r=1}^{\SRN}\kappa_r/\la_r(1+iq^{-\frac{1}{2}}\la_r\kappa_r)(1+iq^{-\frac{1}{2}}\la_r/\kappa_r)\,,\\
& {\tt d}_\SRN^{}(\la)\,=\,(+i)^\SRN\prod_{r=1}^{\SRN}\kappa_r/\la_r(1-iq^{+\frac{1}{2}}\la_r\kappa_r)(1-iq^{+\frac{1}{2}}\la_r/\kappa_r)\,.
\end{aligned}
\end{equation}

{\sc(C) Commutativity}
\begin{equation}
\begin{aligned}
&[\, \SQ(\la)\,,\,\SQ(\mu)\,]\,=\,0\,,\\
&[\, \ST(\la)\,,\,\SQ(\mu)\,]\,=\,0\,,
\end{aligned}\qquad \forall \la,\mu\,.
\end{equation}

{\sc (S) Self-adjointness}\\
Under the assumption $\xi _{r}$ and $\kappa _{r}$ real or imaginary numbers, the following holds:
\begin{equation}
(\ST(\la))^{\dagger}\,=\,\ST(\la^*)\,,\qquad(\SQ(\la))^{\dagger}\,=\,\SQ(\la^*)\,.
\end{equation}
\end{thm}
For the reader's convenience we have
included a self-contained proof 
in Appendix \ref{Qproofs}.

It follows from these properties that $\ST(\la)$ and $\SQ(\mu)$ can be diagonalized simultaneously for all $\la,\mu$.
The eigenvalues $Q(\la)$ of $\SQ(\la)$ must satisfy
\begin{equation}\label{BaxEV}
t(\la)Q(\la)\,=\,{\tt a}_{\SRN}(\la)Q(q^{-1}\la)+{\tt d}_{\SRN}(\la)Q(q\la)\,.
\end{equation}
It follows from the property (A) of $\SQ(\la)$ that any eigenvalue 
$Q(\la)$ must be a polynomial of order $2l\SRN$ 
normalized by the condition $Q_{2l\SRN}=1$.
Such a polynomial is fully characterized
by its zeros $\la_1,\dots,\la_{2l\SRN}$,
\begin{equation}\label{Qfromzeros}
Q(\la)\,=\,\prod_{k=1}^{2l\SRN}(\la-\la_k)\,.
\end{equation}
It follows from the Baxter equation \rf{BaxEV} 
that the zeros must satisfy the
Bethe equations
\begin{equation}
\label{BAE}
\frac{{\tt a}(\la_r)}{{\tt d}(\la_r)}\,=\,-
\prod_{s=1}^{2l\SRN}\frac{\la_s-\la_rq}{\la_s-\la_r/q}\,.
\end{equation}
What is not clear at this stage is if 
for each solution of the Bethe equations 
\rf{BAE} there indeed exists an eigenstate of 
$\ST(\la)$ and $\SQ(\mu)$. In order to show that this is the case
we need a method to construct eigenstates from solutions to \rf{BAE}.
The Separation of Variables method will give us such a construction,
replacing the algebraic Bethe ansatz in the cases we consider.

\subsection{Integrability}

In order to recover the light-cone dynamics discussed in subsection \ref{dyn}, let us temporarily return
to the homogeneous case where $\xi_n=1$ and $\kappa_n=\kappa$ for $n=1,\dots,\SRN$.
Let us note that the operators $\SY(\la)$ simplify when $\la$ is sent to $0$ or $\infty$. Multiplying by
suitable normalization factors one find the {\it unitary} operators
\[
\SY_0\,\equiv\,\ga_{0}^{\SRN}\,\SY(0)\quad{\rm and}\quad
\SY_\infty\,\equiv\,\lim_{\mu\ra\infty}\ga_{\infty}^{\SRN}\,\mu^{-2l\SRN}\SY(\mu)\,,
\]
where $\ga_{0}=\prod_{r=1}^l(1-q^{4r})$ and $\ga_{\infty}=(-1)^{l}q^{l}\prod_{r=1}^l(1-q^{4r-2})$. The operators $\SY_0$ and $\SY_\infty$ have the simple matrix elements
\begin{equation}\label{Y0matel}
\begin{aligned}
\langle\,{\mathbf z}\,|\,\SY_0\,|\,{\mathbf z}'\,\rangle\,& =\,\prod_{n=1}^{\SRN}q^{-2k_n^{}(k_n'+k_{n+1}')}\,,\\
\langle\,{\mathbf z}\,|\,\SY_\infty\,|\,{\mathbf z}'\,\rangle\,&=\,\prod_{n=1}^{\SRN}q^{+2k_n^{}(k_n'+k_{n+1}')}\,,
\end{aligned}\quad{\rm if} \quad \left\{\;\begin{aligned}
&{\mathbf z}\,=\,(q^{2k_1},\dots,q^{2k_{\SRN}}),\\[1ex]
&{\mathbf z}'\,=\,(q^{2k_1'},\dots,q^{2k_{\SRN}'}),
\end{aligned}\;\right\}\end{equation}
and
\begin{equation}
\SQ^+(\la)\,=\,\SY(\la)\cdot\SY_\infty^\dagger\,,\qquad
\SQ^-(\la)\,=\,\big(\SY(\la)\cdot\SY_0^\dagger\,\big)^{-1}
\end{equation}
Integrability follows immediately from the following observation:
\begin{equation}\label{UfromQ}
\boxed{
\qquad\SU\,=\,\alpha_{\kappa} \,\SU^+\cdot\SU^-,\qquad\SU^+\,=\,\SQ^+(1/\kappa\ep),\qquad\SU^-\,=\,\SQ^-(\kappa/\ep),\qquad
}
\end{equation}
where $\alpha_{\kappa}\equiv \prod_{r=1}^{l}(1-q^{4r-2})^{2\SRN}/( \kappa^2-q^{4r-2})^{2\SRN}$.
The proof can be found in Appendix B.
It is very important to remark that there is of course no problem to construct time
evolution operators in the inhomogeneous cases by specializing the spectral parameter of the
$\SQ$-operator in a suitable way. We are just not able to represent the time evolution
as simple as in \rf{Hirota}. One will still have a lattice approximation to the time evolution in the
continuum field theory as long as the inhomogeneity parameters are scaled to unity in the continuum limit.

\section{Separation of variables I --- Statement of results}\label{SOV}

\setcounter{equation}{0}

The Separation of Variables (SOV) method
of Sklyanin \cite{Sk1}-\cite{Sk3} 
as developed for lattice Sine-Gordon model in this section
will allow us to take an important step towards
the simultaneous diagonalization of the $\ST$- and $\SQ$-operators.

The separation of variables method is based on the observation that the
spectral problem for $\ST(\la)$ simplifies considerably if one works 
in an auxiliary
representation where the commutative family 
$\SB(\la)$ of operators introduced in \rf{ABCD} is diagonal.
In the following subsection we will discuss a family of representations
of the Yang-Baxter algebra \rf{YBA}
that has this property. 
We will refer to this class of representations as the SOV-representations.
We will subsequently show that our original
representation introduced in \rf{Mdef}, \rf{Lax} is indeed equivalent to a
certain SOV-representation.

\subsection{The SOV-representation}

The operators representing  \rf{YBA} in the SOV-representation relevant for the case of a lattice with $\SRN$ sites will be denoted as
\begin{equation}
\SM^{\rm\sst SOV}(\la)=\left(\begin{matrix}\SA_\SRN(\la) & \SB_\SRN(\la)\\
\SC_\SRN(\la) & \SD_\SRN(\la)\end{matrix}\right)\,.
\end{equation}
We will now describe the representation
of the algebra \rf{YBA} in which $\SB_\SRN(\la)$ acts diagonally.

\subsubsection{The spectrum of $\SB_\SRN(\la)$}

By definition, we require that $\SB_\SRN(\la)$ is represented by a diagonal matrix. In order to parameterize the eigenvalues,
let us fix a tuple ${\mathbf \zeta}=(\zeta_1^{},\dots,\zeta_\SRN^{})$ of complex numbers such that $\zeta_a^p\neq\zeta_b^p$ for $a\neq b$.
The vector space $\BC^{p^\SRN}$ underlying the SOV-representation will be identified with the space of
functions $\Psi(\eta)$ defined for $\eta$ taken from
the discrete set
\begin{equation}
{\mathbb B_\SRN}\,\equiv\,\big\{\,(q^{k_1}\zeta_1,\dots,q^{k_\SRN}\zeta_\SRN)\,;\,(k_1,\dots,k_\SRN)\in\BZ_p^\SRN\,\big\}\,.
\end{equation}
The SOV-representation is characterized by the property that $\SB(\la)$ acts on the functions $\Psi(\eta)$,
$\eta=(\eta_1,\dots,\eta_\SRN)\in\BB_\SRN$ as a
multiplication operator,
\begin{equation}\label{Bdef}
\SB_\SRN(\la)\,\Psi(\eta)\,=\,\eta_\SRN^{{\rm e}_\SRN}\,b_\eta(\la)\,\Psi(\eta)\,,\qquad b_\eta(\la)\,\equiv\,
\prod_{n=1}^{\SRN}\frac{\kappa _{n}}{i}\prod_{a=1}^{[\SRN]}\left( \la/\eta_a-\eta_a/\la\right)\,;
\end{equation}
where $[\SRN]\equiv\SRN-\en$. We see that $\eta_1,^{}\dots,\eta_{[\SRN]}^{}$ represent the zeros of $b_\eta(\la)$. In the case of
even $\SRN$ it turns out that we need a supplementary variable $\eta_\SRN$ in order to be able to parameterize the spectrum
of $\SB(\la)$.

\subsubsection{Representation of the remaining operators}

Given that $\SB_\SRN(\la)$ is represented as in \rf{Bdef},
it can be shown \cite{Sk1}-\cite{Sk3}\footnote{See \cite{BT}
for the case of the Sinh-Gordon model
which is very similar to  the case at hand.}
that the representation of the remaining
operators $\SA_\SRN(\la)$, $\SC_\SRN(\la)$ $\SD_\SRN(\la)$
is to a large extend determined by
the algebra \rf{YBA}. First note (see e.g. \cite[Appendix C.2]{BT} for a proof) that the so-called
quantum determinant
\begin{equation}\label{qdetdef}
{\rm det_q}(\SM(\la))\,\equiv\,
\SA(\la)\SD(q^{-1}\la)-\SB(\la)\SC(q^{-1}\la)
\end{equation}
generates central elements of the algebra \rf{YBA}. In
the representation defined by \rf{Mdef}, \rf{Lax} we find that
$\la^{2\SRN}{\rm det_q}(\SM(\la))$ is a polynomial in $\la^2$ of order $2\SRN$.
We therefore require that
\begin{equation}\label{detcond}
\SA_\SRN(\la)\SD_\SRN(q^{-1}\la)-\SB_\SRN(\la)\SC_\SRN(q^{-1}\la)\,=\,\Delta_\SRN(\la)\cdot {\rm id}\,,
\end{equation}
with $\la^{2\SRN}\Delta_\SRN(\la)$ being a
polynomial in $\la^2$ of order  $2\SRN$.

The algebra \rf{YBA} furthermore implies that $\SA_\SRN(\la)$ and
$\SD_\SRN(\la)$ can be represented in the form
\begin{align}\label{SAdef}
\SA_\SRN(\la)\,=\,&\,{\rm e}_{\SRN}^{}\,b_\eta(\la)\left[
\frac{\la}{\eta_\SA^{}}\ST^+_\SRN-\frac{\eta_\SA^{}}{\la}\ST^-_\SRN\right]
+\sum_{a=1}^{[\SRN]}\prod_{b\neq a}\frac{\la/\eta_b-\eta_b/\la}{\eta_a/\eta_b-\eta_b/\eta_a} \,a_\SRN^{}(\eta_a)\,\ST_a^-\,,\\
\SD_\SRN(\la)\,=\,&\,{\rm e}_{\SRN}^{}\,b_\eta(\la)\left[\frac{\la}{\eta_\SD^{}}\ST^-_\SRN-\frac{\eta_\SD^{}}{\la}\ST^+_\SRN\right]
+\sum_{a=1}^{[\SRN]}\prod_{b\neq a}\frac{\la/\eta_b-\eta_b/\la}{\eta_a/\eta_b-\eta_b/\eta_a} \,d_\SRN^{}(\eta_a)\,\ST_a^+\,,
\label{SDdef}\end{align}
where $\ST_a^{\pm}$ are
the operators defined by
\[
\ST_a^{\pm}\Psi(\eta_1,\dots,\eta_\SRN)=\Psi(\eta_1,\dots,q^{\pm 1}\eta_a,\dots,\eta_\SRN)\,.
\]
The expressions  \rf{SAdef} and \rf{SDdef} contain complex-valued
coefficients  $\eta_\SA^{}$, $\eta_\SD^{}$, $a_\SRN(\eta_r)$ and
$d_\SRN(\eta_r)$.
The coefficients $a_\SRN(\eta_r)$ and
$d_\SRN(\eta_r)$ are restricted  by the condition
\begin{equation}\label{addet}
\Delta_\SRN(\eta_r)\,=\,
a_\SRN(\eta_r)d_\SRN(q^{-1}\eta_r)\,, \quad\forall r=1,\dots,\SRN\,,
\end{equation}
as follows from the consistency of \rf{detcond}, \rf{Bdef},
\rf{SAdef} and \rf{SDdef}. This leaves some freedom in the
choice of $a_\SRN(\eta_r)$ and
$d_\SRN(\eta_r)$ that will be further discussed later.

The operator $\SC_\SRN(\la)$ is finally univocally\footnote{Note that the operator $\SB_\SRN(\la)$ is invertible except for $\la$ which coincides with a zero of $\SB_\SRN$, so in general $\SC_\SRN(\la)$ is defined by (4.5) just inverting $\SB_\SRN(\la)$. This is enough to fix in an unique way the operator $\SC_\SRN$ being it a Laurent polynomial of degree [$\SRN$] in $\la$.} defined such that the
quantum determinant condition \rf{detcond}
is satisfied.

\subsubsection{Central elements} \label{center}

For the representations in question, the algebra \rf{YBA} has a large
center. For its description let us,
following \cite{Ta}, define the average value $\CO$ of the elements
of the monodromy matrix $\SM^{\rm\sst SOV}(\la)$ as
\begin{equation}\label{avdef}
\CO(\Lambda)\,=\,\prod_{k=1}^{p}\SO(q^k\la)\,,\qquad \Lambda\,=\,\la^p,
\end{equation}
where $\SO$ can be
$\SA_\SRN$, $\SB_\SRN$,
$\SC_\SRN$ or $\SD_\SRN$.

\begin{propn}
The average values $\CA_\SRN(\Lambda)$, $\CB_\SRN(\Lambda)$,
$\CC_\SRN(\Lambda)$, $\CD_\SRN(\Lambda)$ of the monodromy matrix
$\SM(\la)$ elements are central elements.
\end{propn}

The Proposition is proven in \cite{Ta}, see Subsection \ref{Avvalapp}
for an alternative proof.
The average values are of course unchanged by
similarity transformations.
They therefore represent
parameters of the representation. Let us briefly discuss how these parameters
are related to the parameters of the SOV-representation introduced above.

First, let us note that $\CB_\SRN(\Lambda)$ is easily found from \rf{Bdef}
to be given by the formula
\begin{equation}\label{CB}
\CB_\SRN(\Lambda)\,=\,Z_{{\SRN}}^{{\rm e}_{\SRN}}
\prod_{n=1}^{\SRN}\frac{K_n}{i^p}\prod_{a=1}^{{[}\SRN{]}}(\Lambda/Z_a-Z_a/\Lambda)\,,\qquad
\begin{aligned}
& Z_a\equiv \eta_a^p\,,\\
& K_a\equiv\kappa_a^p\,.
\end{aligned}
\end{equation}

The values  $\CA_\SRN^{}(Z_r)$ and
$\CD_\SRN^{}(Z_r)$ are related to the coefficients $a_\SRN(q^k\eta_r)$
and $d_\SRN(q^k\eta_r)$ by
\begin{equation}\label{ADaver}
\CA_\SRN^{}(Z_r)\,\equiv\,\prod_{k=1}^{p}a_\SRN(q^k\eta_r)\,,\qquad
\CD_\SRN^{}(Z_r)\,\equiv\,\prod_{k=1}^{p}d_\SRN(q^k\eta_r)\,.
\end{equation}
Note that the condition \rf{addet} leaves
some remaining arbitrariness in the choice of the coefficients $a_\SRN(\eta)$,
$d_\SRN(\eta)$. The gauge transformations
\begin{equation}
\Psi(\eta)\,\equiv\,\prod_{r=1}^{\SRN}f(\eta_r)\Psi'(\eta)\,,
\end{equation}
induce a change of coefficients
\begin{equation}\label{gauge}
a_\SRN'(\eta_r)\,=\,a_\SRN(\eta_r)\frac{f(q^{-1}\eta_r)}{f(\eta_r)}\,,
\qquad
d_\SRN'(\eta_r)\,=\,d_\SRN(\eta_r)\frac{f(q^{+1}\eta_r)}{f(\eta_r)}\,,
\end{equation}
but clearly leave $\CA_\SRN^{}(Z_r)$ and $\CD_\SRN^{}(Z_r)$ unchanged. The data
$\CA_\SRN^{}(Z_r)$ and $\CD_\SRN^{}(Z_r)$ therefore characterize gauge-equivalence classes
of representations for $\SA_\SRN(\la)$ and $\SD_\SRN(\la)$ in the form \rf{SAdef}.

\subsection{Existence of SOV-representation for the lattice Sine-Gordon model}
\label{SOVex}

We are looking for an invertible transformation $\SW^{\rm\sst SOV}$ that maps the lattice Sine-Gordon model defined in
the previous sections to a SOV-representation,
\begin{equation}\label{inter}
(\SW^{\rm\sst SOV})^{-1}\cdot\SM^{\rm\sst SOV}(\la)\cdot\SW^{\rm\sst SOV}\,=\,\SM(\la)\,.
\end{equation}
Constructing $\SM^{\rm\sst SOV}(\la)$ is of course equivalent to the construction of a
basis for $\CH$ consisting of eigenvectors $\langle\,\eta\,|$ of $\SB(\la)$,
\begin{equation}
\langle\,\eta\,|\,\SB(\la)\,=\,\eta_\SRN^{{\rm e}_\SRN}b_\eta(\la)\,\langle\,\eta\,|\,.
\end{equation}
The transformation $\SW^{\rm\sst SOV}$ is then described in terms of $\langle \,\eta\,|\,z\,\rangle$ as
\begin{equation}\label{Ws}
(\SW^{\rm\sst SOV}\psi)(\eta)\,=\,\sum_{z\in(\BS_p)^{\SRN}} \,\langle \,\eta\,|\,z\,\rangle\,\psi(z)\,.
\end{equation}

The existence of an eigenbasis for $\SB(\la)$
is not trivial since $\SB(\la)$ is not a normal operator.
It turns out that such a similarity transformation exists
for generic values of the parameters $u$, $v$, $\xi$ and $\kappa $.

\begin{thm} \label{SOVthm} $\;\;$ -- $\;\;${\bf Existence of SOV-representation for the lattice Sine-Gordon model}$\;\;$ --\\[1ex]
For generic values of the parameters $u$, $v$, $\xi$ and $\kappa$
there exists an invertible operator
$\SW^{\rm\sst SOV}:\CH\ra\CH^{\rm\sst SOV}$ which satisfies \rf{inter}.
\end{thm}

The proof is given in the following Section \ref{SOVapp}.
It follows from \rf{SAdef}, \rf{SDdef} that the
wave-functions $\Psi(\eta)=\langle\,\eta\,|\,t\,\rangle$
of eigenstates $|\,t\,\rangle$
must satisfy the discrete Baxter equations
\begin{equation}\label{SOVBax1}
t(\eta_n)\Psi(\eta)\,=\,a(\eta_n)\ST_n^-\Psi(\eta)
+d(\eta_n)\ST_n^+\Psi(\eta)\,,
\end{equation}
where $n=1,\dots,\SRN$.
Equation \rf{SOVBax1} represents a system of $p^\SRN$ linear
equations for the $p^\SRN$ different components $\Psi(\eta)$ of the
vector $\Psi$. It may be written in the form ${D}_t\cdot\Psi=0$,
where ${D}_t$ is a $p^\SRN\times p^\SRN$-matrix
that depends on $t=t(\la)$.
The condition for existence of solutions
${\rm det}{D}_t=0$ is a polynomial equation
of order $p^\SRN$ on $t(\la)$. We therefore expect to find
$p^\SRN$ different solutions, just
enough to get a basis for $\CH$.

We will return to the analysis of the spectral problem of $\ST(\la)$ in
Section \ref{Spec}. Let us now describe more precisely the
set of values of the parameters for which a SOV-representation
exists.

\subsection{Calculation of the average values}

Necessary condition for the existence of $\SW^{\rm\sst SOV}$ is of
course the equality
\begin{equation}\label{CM=CM}
\CM(\Lambda)\,=\,\CM^{\rm\sst SOV}(\Lambda)\,,
\end{equation}
of the matrices formed out of the average values of
$\SM(\lambda)$ and $\SM^{\rm\sst SOV}(\lambda)$, respectively.
It turns out that $\CM(\Lambda)$ can be calculated recursively from the
average values of the elements of the Lax
matrices $L_n^{\rm\sst SG}(\la)$, which are explicitly given by
\begin{align}\label{L_n}
\CL_n(\Lambda)& 
\,=\,\frac{1}{i^p}
\left(
\begin{matrix} i^p U_n^{}(K_n^{2}V_n^{}+V_n^{-1}) &
K_n(\Lambda V_n/X_n-X_n/V_n\Lambda) \\
K_n(\Lambda/ X_nV_n-X_nV_n/\Lambda) &
i^p U_n^{-1}(K_n^{2}V_n^{-1}+V_n^{})
\end{matrix}
\right)\,,
\end{align}
where we have used the notations
$K_n=\kappa_n^p$,  $X_n=\xi_n^p$,
$U_n=u_n^p$ and $V_n=v_n^p$. Indeed, we have:
\begin{propn}\label{Avrec}
We have \begin{align}\label{RRel1a}
\CM_{\SRN}^{}(\Lambda)\,=\,
\CL_{\SRN}^{}(\Lambda)\,\CL_{\SRN-1}^{}(\Lambda)\,\dots\,\CL_1^{}(\Lambda)\,.
\end{align}
\end{propn}
This has been proven in \cite{Ta},
see Subsection \ref{Avvalapp} for an alternative proof.

The equality \rf{CM=CM} defines the mapping between the parameters
$u,v,\kappa$ and $\xi$ of the representation defined in
Subsection \ref{T-op}  and the parameters
of  the SOV-representation. Formula \rf{RRel1a}
in particular allows us to calculate $\CB(\Lambda)$
in terms of $u,v,\kappa$ and $\xi$. Equation \rf{CB}
then defines the numbers
$Z_a\equiv\eta_a^p$
uniquely up to permutations of $a=1,\dots,[\SRN]$.

Existence of a SOV-representation in particular requires that
$Z_a\neq  Z_b$ for all $a\neq b$, $a,b=1,\dots,[\SRN]$.
It can be shown (see Subsection \ref{nondegapp} below)
that the subspace of the space of parameters $u$, $v$,
$\kappa$ and $\xi$ for which this is not the case has codimension at
least one. Sufficient
for the existence of a SOV-representation is
the condition that
the representations $\CR_\SRM$ exist
for all $\SRM=1,\dots,\SRN-1$.

\section{Separation of variables II --- Proofs}\label{SOVapp}

\setcounter{equation}{0}

We are now proving Theorem \ref{SOVthm} by constructing a
set of $p^N$ linearly independent vectors $\langle\,\eta\,|$
which are eigenvectors of $\SB(\la)$ with distinct eigenvalues. This will be equivalent to a recursive construction of the matrix of elements $\langle \,\eta\,|\,z\,\rangle$ and so of the invertible operator $\SW^{\rm\sst SOV}:\CH\ra\CH^{\rm\sst SOV}$ by relation \rf{Ws}.

\subsection{Construction of an eigenbasis for $\SB(\la)$}

We will construct the eigenstates $\langle\,\eta\,|$ of $\SB(\la)\equiv \SB_\srn(\la)$
recursively by induction on $\SRN$. The corresponding eigenvalues $B(\la)$ are parameterized by the tuple
$\eta=(\eta_a)^{}_{a=1,\dots,\SRN}$ as
\begin{equation}\label{Bparam}
B(\la)\,=\,\eta_\SRN^{{\rm e}_\SRN}\,b_\eta(\la)\,,\qquad b_\eta(\la)\,\equiv\,
\prod_{n=1}^{\SRN}\frac{\kappa _{n}}{i}\prod_{a=1}^{[\SRN]}\left( \la/\eta_a-\eta_a/\la\right)\,;
\end{equation}

We remind that $e_\srn$ is zero for $\SRN$ odd and $1$ for $\SRN$ even.

In the case $\SRN=1$ we may
simply take $\langle\,\eta_1\,|\,=\,\langle\, v\,|, $ where  $\langle\, v\,|$ is an eigenstate of the operator $\sv_1$ with eigenvalue $v$.
It is useful to note that the inhomogeneity parameter determines the subset of $\BC$ on which the
variable $\eta_1$ lives, $\eta_1\in\xi_1\BS_p$.

Now assume we have constructed the eigenstates $\langle\,\chi\,|$ of $\SB_{\srm}(\la)$ for any $\SRM<\SRN$. The eigenstates
$\langle\,\eta\,|$, $\eta=(\eta_\srn,\dots,\eta_1)$, of $\SB_{\srn}(\la)$ may then be constructed in the following form
\begin{equation}
\langle\,\eta\,|\,=\,\sum_{\chi_{\1}^{}}\sum_{\chi_\2^{}}\,K_\srn^{}(\,\eta\,|\,\chi_{\2}^{};\chi_\1^{}\,)\,
\langle \,\chi_{\2}^{}\,|\ot\langle\,\chi_\1^{}\,|\;,
\end{equation}
where $\langle \,\chi_{\2}^{}\,|$ and $\langle \,\chi_{\1}^{}\,|$ are eigenstates of $\SB_\srm(\la)$ and $\SB_{\srn-\srm}(\la)$
with eigenvalues parameterized as in \rf{Bparam}  by the tuples $\chi_\2^{}=(\chi_{\2 a}^{})_{a=1,\dots ,\srm}^{}$
and $\chi_\1^{}=(\chi_{\1 a}^{})_{a=1,\dots,\srn-\srm}^{}$, respectively. It suffices to consider the cases where $\SRN-\SRM$ is
odd.

It follows from the formula
\begin{equation}\begin{aligned}
\SB_\srn(\la)&\,=\,\SA_\srm(\la)\ot\SB_{\srn-\srm}(\la)+{\SB}_\srm(\la)\ot\SD_{\srn-\srm}(\la)\\
&\,\equiv\,\SA_{\2 \ \srm}(\la)\SB_{\1 \ \srn-\srm}(\la)+{\SB}_{\2 \ \srm}(\la)\SD_{\1 \ \srn-\srm}(\la)
\end{aligned}
\end{equation}
that the matrix elements $K_\srn(\,\eta\,|\,{\chi}_{\2}^{};{\chi}_\1^{}\,)$ have to satisfy the relations
\begin{equation}\begin{aligned}\label{recrel}
\big(\SA_{\2 \ \srm}(\la)\SB_{\1 \ \srn-\srm}(\la) +{\SB}_{\2 \ \srm}(\la)& \SD_{\1 \ \srn-\srm}(\la)\big)^t\,K_\srn^{}(\,\eta\,|\,\chi_{\2}^{};\chi_\1^{}\,)
\\
&\,=\,\eta_\srn^{{\rm e}_\srn}\prod_{n=1}^{\SRN}\frac{\kappa _{n}}{i}\prod_{a=1}^{[\SRN]}\left( \la/\eta_a-\eta_a/\la\right)\,K_\srn^{}(\,\eta\,|\,\chi_{\2}^{};\chi_{\1}^{}\,)\,,
\end{aligned}\end{equation}
where we used the notation $\SO^t$ for the transpose of an operator $\SO$.

Let us assume that
\begin{equation}
\chi_{\1a}q^{h_{1}}\notin \Delta _{\1}, \ \ \chi _{\2b}q^{h_{2}}\notin \Delta
_{\2}\ \ \text{and} \ \ \chi _{\1a}q^{h_{1}}\neq \chi _{\2b}q^{h_{2}},
\end{equation}
where $h_{i}\in \{1,...,p\}$, $a\in \{1,...,\SRN-\SRM\}$\ , $b\in \{1,...,\SRM\}$\ and $\Delta _{i}$ is the set of zeros of the quantum
determinant on the subchain $i$, with $i=\1$,$\2$. Under these assumptions\footnote{%
The subspace within the space of parameters where these
conditions are not satisfied has codimension at least one.} the previous equations yield recursion relations for the dependence of the kernels in
the variables $\chi_{\1 a}$ and $\chi_{\2 b}$ simply by setting $\la=\chi_{\1 a}$ and $\la=\chi_{\2 b}$. Indeed for $\la=\chi_{\1 a}$
the first term on the left of \rf{recrel} vanishes leading to
\begin{equation}\label{recrel1}\begin{aligned}
 {\ST_{\1 a}^{^-}
 K_\srn^{}(\,\eta\,|\,\chi_{\2}^{};\chi_{\1}^{}\,)}\;&{d_\1^{}(q^{-1}\chi_{\1a}^{})}\; \,\chi_\srm^{{\rm e}_\srm}\prod_{n=1}^{\SRN-\SRM}
\frac{i}{\kappa _{n}}\prod_{a=1}^{[\rm M]}(\chi_{\1 a}^{}/\chi_{\2 b}^{}-\chi_{\2 b}^{}/\chi_{\1 a}^{})\, \\[-1ex]
& \,=\, {K_\srn^{}(\,\eta\,|\,\chi_{\2}^{};\chi_{\1}^{}\,)}\; \,\eta_\srn^{{\rm e}_\srn}\prod_{b=1}^{[\SRN]}(\chi_{\1 a}^{}/\eta_b^{}-\eta_b^{}/\chi_{\1 a}^{})
\,,
\end{aligned}\end{equation}
while for $\la=\chi_{\2a}$ one finds similarly
\begin{equation}\label{recrel2}\begin{aligned}
{\ST_{\2 a}^{^+}
 K_\srn^{}(\,\eta\,|\,\chi_{\2}^{};\chi_{\1}^{}\,)}\; &{a_\2^{}(q^{+1}\chi_{\2a}^{})}\;
 \prod_{n=1}^{\rm M}\frac{i}{\kappa _{n}} \prod_{b=1}^{\SRN-\rm M}(\chi_{\2 a}^{}/\chi_{\1 b}^{}-\chi_{\1 b}^{}/\chi_{\2 a}^{} )\,
\\[-1ex] &\,=\,{K_\srn^{}(\,\eta\,|\,\chi_{\2}^{};\eta_{\1}^{}\,)}\; \,\eta_\srn^{{\rm e}_\srn}\prod_{b=1}^{[\SRN]}
(\chi_{\2 a}^{}/\eta_b^{}-\eta_b^{}/\chi_{\2 a}^{})
\,.
\end{aligned}\end{equation}
If ${\rm M}$ is even we find the recursion relation determining the dependence on $\chi_{\2\srm}$ by sending $\la\ra\infty$
in \rf{recrel}, leading to
\begin{equation}\label{recrelzero}
{\ST_{\2 \srm}^{^+}
 K_\srn^{}(\,\eta\,|\,\chi_{\2}^{};\chi_{\1}^{}\,)}\; \frac{1}{\chi_{\2 \SA}^{}}\,\prod_{a=1}^{\rm M-1}\frac{1}{\chi_{\2 a}^{}}\,\prod_{b=1}^{\SRN-\rm M}\frac{1}{\chi_{\1 b}^{}}\,
\,=\,{K_\srn^{}(\,\eta\,|\,\chi_{\2}^{};\eta_{\1}^{}\,)}\;
\prod_{b=1}^{\SRN}\frac{1}{\eta_b^{}}
\,.
\end{equation}
The recursion relations \rf{recrel1}, \rf{recrel2} have solutions compatible with
the requirement of cyclicity,  $ ({\ST_{\1 a}^{^-}})^p=1$ and $({\ST_{\2 a}^{^+})^p=1}$ for all values of $a$, provided that the algebraic equations
\begin{equation}\label{Exrecrel1}\begin{aligned}
& {D_\1^{}(\chi_{\1a}^{})}\; (\chi_{\2\srm}^{{\rm e}_\srm})^{p}\prod_{n=1}^{\SRN-\SRM}
\frac{i^p}{\kappa _{n}^p}\prod_{b=1}^{[{\rm M}]}
 (\chi_{\1 a}^{p}/\chi_{\2 b}^{p}-\chi_{\2 b}^{p}/\chi_{\1 a}^{p})\,
\,=\, (\eta_\srn^{{\rm e}_\srn})^{p}\prod_{b=1}^{[\SRN]}(\chi_{\1 a}^{p}/\eta_b^{p}-\eta_b^{p}/\chi_{\1 a}^{p})\,,\\
& {\rm where}\;\;D_\1(\chi_{\1 a})\,
\equiv\,\prod_{k=1}^p d_\1(q^k\chi_{\1 a})\,,
\end{aligned}\end{equation}
and
\begin{equation}\label{Exrecrel2}\begin{aligned}
&{A_\2^{}(\chi_{\2 a})}\; \prod_{n=1}^{\SRM}
\frac{i^p}{\kappa _{n}^p}
 \prod_{b=1}^{\SRN-\rm M}(\chi_{\2 a}^{p}/\chi_{\1 b}^{p}-\chi_{\1 b}^{p}/\chi_{\2 a}^{p} )\,=\,
 (\eta_\srn^{{\rm e}_\srn})^{p}\prod_{b=1}^{[\SRN]}(\chi_{\2 a}^{p}/\eta_b^{p}-\eta_b^{p}/\chi_{\2 a}^{p})\,,\\
& {\rm where}\;\;A_\2(\chi_{\2 a})\,\equiv\,\prod_{k=1}^p a_\2(q^k\chi_{\2 a}) \,,
\end{aligned}
\end{equation}
are satisfied.  If $\SRM$ is even the recursion relation \rf{recrelzero}
yields the  additional relation
\begin{equation}\label{Exrecrel3}
\frac{1}{\chi_{\2 \SA}^{p}}\,\prod_{a=1}^{\rm M-1}\frac{1}{\chi_{\2 a}^{p}}\,
\prod_{b=1}^{\SRN-\rm M}\frac{1}{\chi_{\1 b}^{p}}\,
\,=\,
\prod_{b=1}^{\SRN}\frac{1}{\eta_b^{p}}
\,.
\end{equation}
We will show in the next subsection that the equations \rf{Exrecrel1}-\rf{Exrecrel3} completely determine
$\eta_a^p$ in terms of $\chi_{\2 a}^p$, $\chi_{\1 a}^p$.

By using \rf{CB} and \rf{AADD} it is easy to see
that the conditions \rf{Exrecrel1} and \rf{Exrecrel2}
are nothing but the equations
\begin{equation}\label{Brecrel}
\mathcal{B}_{\SRN}(\Lambda ) =\mathcal{A}_{\SRM}(\Lambda )\mathcal{B}%
_{\SRN-\SRM}(\Lambda )+\mathcal{B}_{\SRM}(\Lambda )\mathcal{D}_{\SRN-\SRM}(\Lambda ),
\end{equation}
evaluated at $\Lambda=\chi_{\1a}^p$ and $\Lambda=\chi_{\2a}^p$, respectively.
The relation \rf{Exrecrel3} follows from \rf{Brecrel} in the limit
$\la\ra\infty$. The relations \rf{Brecrel}
are implied by \rf{RRel1a}. We conclude that our construction
of $\SB(\la)$-eigenstates will work if the representations
$\CR_{\SRN}$, $\CR_{\SRM}$ and $\CR_{\SRN-\SRM}$
are all non-degenerate. Theorem 2 follows by induction.

\subsection{On average value formulae}\label{Avvalapp}

\begin{propn}
The average values of the Yang-Baxter generators are central elements which
satisfy the following recursive equations:
\begin{eqnarray}
\mathcal{B}_{\SRN}(\Lambda ) &=&\mathcal{A}_{\SRM}(\Lambda )\mathcal{B}%
_{\SRN-\SRM}(\Lambda )+\mathcal{B}_{\SRM}(\Lambda )\mathcal{D}_{\SRN-\SRM}(\Lambda ),
\label{average value-B} \\
\mathcal{C}_{\SRN}(\Lambda ) &=&\mathcal{D}_{\SRM}(\Lambda )\mathcal{C}%
_{\SRN-\SRM}(\Lambda )+\mathcal{C}_{\SRM}(\Lambda )\mathcal{A}_{\SRN-\SRM}(\Lambda ),
\label{average value-C} \\
\mathcal{A}_{\SRN}(\Lambda ) &=&\mathcal{A}_{\SRM}(\Lambda )\mathcal{A}%
_{\SRN-\SRM}(\Lambda )+\mathcal{B}_{\SRM}(\Lambda )\mathcal{C}_{\SRN-\SRM}(\Lambda ),
\label{average value-A} \\
\mathcal{D}_{\SRN}(\Lambda ) &=&\mathcal{D}_{\SRM}(\Lambda )\mathcal{D}%
_{\SRN-\SRM}(\Lambda )+\mathcal{C}_{\SRM}(\Lambda )\mathcal{B}_{\SRN-\SRM}(\Lambda ),
\label{average value-D}
\end{eqnarray}
where $\SRN-\SRM$ or $\SRM$ is odd.
\end{propn}
\begin{proof}
In the previous subsection
we have proven the existence of SOV-representations, i.e. the
diagonalizability of the $\SB$-operator.
First of all let us point out that $\SA(\lambda )$, $\SB(\lambda )$, $\SC(\lambda )$ and $\SD(\lambda )$
are one parameter families of commuting operators. This implies that
the corresponding average values are functions of $\Lambda =\lambda ^{p}$.

The fact that $\CB_\SRN(\Lambda)$ is central trivially follows from the fact that
$\SB_\SRN(\la)$ is diagonal in the SOV-representation, while for the operators
$\SA$ and $\SD$ we have that for $\SRN$ odd, $\mathcal{A}_{\SRN}(\Lambda )\Lambda ^{\SRN-1}$ and $\mathcal{D}_{\SRN}(\Lambda )\Lambda ^{\SRN-1}$ are polynomials in $\Lambda
^{2} $ of degree $\SRN-1$. It follows that the special values
given by \rf{ADaver}
characterize them
completely,
\begin{equation}
\begin{aligned}
\mathcal{A}_{\SRN}(\Lambda )&\,=\,
\sum_{a=1}^{[\SRN]}\prod_{b\neq a}\frac{(\Lambda/Z _{b}-Z _{b}/\Lambda )}{(Z_{a}/Z_{b}-Z _{b}/Z_{a})}\emph{A}_{\SRN}(Z_{a})\,,\\
\mathcal{D}_{\SRN}(\Lambda )&\,=\,
\sum_{a=1}^{[\SRN]}\prod_{b\neq a}
\frac{(\Lambda/Z_{b}-Z_{b}/\Lambda)}{(Z_{a}/Z_{b}-Z_{b}/Z_{a})}
\emph{D}_{\SRN}(Z_{a}),
\end{aligned}
\end{equation}
where $\emph{A}_{\SRN}(Z_{a})$ and $\emph{D}_{\SRN}(Z_{a})$ are the average values of the coefficients of the SOV-representation.
In the case of $\SRN$ even we have just to add the asymptotic
property of $\mathcal{A}_{\SRN}(\Lambda )$ and $\mathcal{D}_{\SRN}(\Lambda )$\
discussed in appendix \ref{Asymp-A-D} to complete the statement. Finally, the fact that $%
\mathcal{C}_{\SRN}(\Lambda )$ is central follows by its diagonalizability in
the cyclic representations.

Now the above recursive formulae (\ref{average value-B}-\ref{average value-D}) are
a simple consequence of the centrality of the average values of the monodromy matrix elements. Let us consider only the case of the average value of $\SA_{\SRN}(\lambda )$. We have the expansion:
\begin{equation}
\SA_{\SRN}(\lambda )=\SA_{\2 \ \SRM}(\lambda )\SA_{\1 \ \SRN-\SRM}(\lambda )+\SB_{\2 \ \SRM}(\lambda
)\SC_{\1 \ \SRN-\SRM}(\lambda ),  \label{A-expan}
\end{equation}
in terms of the entries of the monodromy matrix of the subchains $\1$ and $\2$
with $(\SRN-\SRM)$-sites and $\SRM$-sites, respectively.
It follows directly from definition \rf{avdef} of the
average value together with \rf{A-expan} that $\mathcal{A}_{\SRN}(\Lambda )$
can be represented in the form
\begin{equation}
\mathcal{A}_{\SRN}(\Lambda )=\mathcal{A}_{\2 \ \SRM}(\Lambda )\mathcal{A}_{\1 \ \SRN-\SRM}(\Lambda
)+\mathcal{B}_{\2 \ \SRM}(\Lambda )\mathcal{C}_{\1 \ \SRN-\SRM}(\Lambda )+\Delta _{\SRN}(\lambda )
\label{Central-D}
\end{equation}
where $\Delta _{\SRN}(\lambda )$ is a sum over monomials which contain
at least one and at most $p-2$ factors of
$\SA_{\2 \ \SRM}(\la q^m)$. As before, we may work in a representation
where the $\SB_{\2 \ \SRM}(\la q^n)$ are diagonal, spanned by the states
$\langle\,\chi_\2\,|$ introduced in the previous subsection. As
the factors $\SA_{\2 \ \SRM}(\la q^m)$ contained in $\Delta _{\SRN}(\lambda )$
produce states with modified
eigenvalue of $\SB_{\2 \ \SRM}(\la q^n)$, none of
the states produced by acting with $\Delta _{\SRN}(\lambda )$ on
$\langle\,\chi_\2\,|$ can be proportional to
$\langle\,\chi_\2\,|$. This would be in contradiction to the fact that
$\mathcal{A}_{\SRN}(\Lambda )$ is central unless $\Delta _{\SRN}(\lambda )=0$.
\end{proof}

\subsection{Non-degeneracy condition}\label{nondegapp}


\begin{propn}
\label{B-simplicity}The condition $Z_{r}=Z_{s}$ for certain $r\neq s$ with $%
r,s\in \{1,...,[\SRN]\}$ defines a subspace  in the space of the parameters $%
\{\kappa_{1},...,\kappa_{\SRN},\xi_{1},...,\xi_{\SRN}\}\in
\mathbb{C}^{2\SRN}$ of codimension at least one.
\end{propn}

\begin{proof}
The parameters $Z_r$ are related to the expectation
value $\CB_\SRN(\Lambda)$ by means of the equation
\begin{equation}\label{CB'}
\CB_\SRN(\Lambda)\,=\,Z_{{\SRN}}^{{\rm e}_{\SRN}}
\prod_{n=1}^{\SRN}\frac{K_n}{i^p}
\prod_{a=1}^{{[}\SRN{]}}(\Lambda/Z_a-Z_a/\Lambda)\,.
\end{equation}
It follows from \rf{RRel1a} and \rf{L_n} that
$\CB_\SRN(\Lambda)$ is a Laurent polynomial in $X_n$
that depends polynomially on each of the parameters $K_n$.
Equation \rf{CB'} defines the tuple $Z=(Z_1^{},\dots,Z_{[\SRN]}^{})$
uniquely up to permutations of $Z_1^{},\dots,Z_{[\SRN]}^{}$
as function of the
parameters $X=(X_1,\dots,X_\SRN)$ and $K=(K_1,\dots,K_\SRN)$.
We are going to show that\footnote{It should be noted that for even
$\SRN$ it is indeed sufficient to consider the dependence
w.r.t. $X_1,\dots,X_{{\SRN}-1}$.}
\begin{equation}
J(X;K)\,\equiv\,
{\rm det}\left(\frac{\pa Z_r}{\pa X_s}\right)_{r,s=1,\dots,[\SRN]}\neq\,0\,.
\end{equation}
The functional dependence\footnote{Let $\sigma_n^{{[\SRN]}}(Z)$ be the degree n elementary symmetric polynomial in the variables Z, then $\sigma_n^{{[\SRN]}}(Z)/\sigma_{[\SRN]}^{[\SRN]}(Z)$ are Laurent polynomials of degree 1 in all the parameters X and K.} of the $Z_1^{},\dots,Z_{[\SRN]}^{}$ w.r.t. the parameters $K$ implies that it is sufficient to show that $J(X;K)\neq 0$
for special values of $K$ in order to prove that $J(X;K)\neq 0$ except for values of $K$ within a subset of
$\BC^\SRN$ of dimension less than $\SRN$.

Let us choose $K_{n}=i^p$ for     $n=1,...,[\SRN]$, then the average values \rf{L_n} of the Lax operators simplify to
\begin{equation}
\mathcal{L}_{n}^{SG}(\Lambda )=\left(
\begin{array}{cc}
0 & \Lambda /X_{n}-X_{n}/\Lambda \\
\Lambda /X_{n}-X_{n}/\Lambda & 0%
\end{array}
\right) \,.
\end{equation}
Inserting this into \rf{RRel1a} yields
\begin{equation}\label{CBsimp}
\mathcal{{B}}_{\SRN}(\Lambda )\,=\,(K_\SRN^2+1)^{{\rm e}_\SRN}
\prod_{n=1}^{[\SRN]}(\Lambda /X_{n}-X_{n}/\Lambda )\,.
\end{equation}
The fact that $J(X;K)\neq 0$ follows for the case under consideration
easily from \rf{CBsimp}.

Whenever $J(X;K)\neq 0$, we have invertibility of the mapping
$Z=Z(X_1,\dots,X_{[\SRN]})$. The claim follows from this observation.
\end{proof}

\section{The spectrum --- odd number of sites}\label{Spec}

\setcounter{equation}{0}

Let us now return to the analysis of the spectrum of the model.
For simplicity we will consider here the case of odd $\SRN$, while we will
discuss the case of even $\SRN$ in the next section. 
The existence of the SOV-representation
allows one to reformulate the spectral problem for $\ST(\la)$
as the problem
to find all solutions of the {\it discrete} Baxter equations
\rf{SOVBax1}. This equation may be written in the form
\begin{equation}\label{SOVBax}
\CD_r\,\Psi(\eta)\,=\,0\,,\qquad
\CD_r\equiv\,a(\eta_r)\ST_r^-+d(\eta_r)\ST_r^+-t(\eta_r)\,,
\end{equation}
where $r=1,\dots,\SRN$. Previous experience with the
SOV method suggests to consider the ansatz
\begin{equation}\label{Qeigenstate1}
\Psi(\eta)=\prod_{r=1}^{\SRN}Q_t(\eta_r)\,,
\end{equation}
where $Q_t(\la)$ is the eigenvalue of the corresponding $\SQ$-operator
which satisfies the {\it functional}
Baxter equations
\begin{equation}\label{BaxEV2a}
\begin{aligned}
&t(\la)Q_t(\la)\,=\,{\tt a}_\SRN(\la)Q_t(q^{-1}\la)+
{\tt d}_\SRN(\la)Q_t(q\la)\,.
\end{aligned}\end{equation}
This approach will turn out to work, but in a way that
is more subtle than in previously analyzed
cases.

\subsection{States from solutions of the Baxter equation}

First, in the present case it is not immediately clear if
the functional Baxter equation \rf{BaxEV2a} and the
discrete Baxter equation \rf{SOVBax} are compatible.
The question is if one can always assume that
the coefficients $a(\eta_r)$ and $d(\eta_r)$ in \rf{SOVBax}
coincide with the coefficients ${\tt a}_\SRN(\eta_r)$,
${\tt d}_\SRN(\eta_r)$ appearing in the functional equation \rf{BaxEV2a}
satisfied by the $\SQ$-operator.
The key point to observe is contained in the following Lemma.
\begin{lem}\label{A=A-B}
Let ${\tt A}_\SRN^{}(\Lambda)$ and ${\tt D}_\SRN^{}(\Lambda)$ be the average values of the coefficients ${\tt a}_\SRN(\la)$
and ${\tt d}_\SRN(\la)$ of the Baxter equation \rf{BaxEV2a},
\begin{equation}
{\tt A}_\SRN^{}(\Lambda)\,\equiv\,\prod_{k=1}^{p}{\tt a}_\SRN(q^k\lambda)\,,\qquad
{\tt D}_\SRN^{}(\Lambda)\,\equiv\,\prod_{k=1}^{p}{\tt d}_\SRN(q^k\lambda)\,.
\end{equation}
We then have
\begin{equation}
{\tt A}_\SRN^{}(\Lambda)\,=\,\CA_\SRN(\Lambda)-\CB_\SRN(\Lambda)\,,\qquad
{\tt D}_\SRN^{}(\Lambda)\,=\,\CA_\SRN(\Lambda)+\CB_\SRN(\Lambda)\,.
\end{equation}
\end{lem}
\begin{proof} The claim is checked for $\SRN=1$ by straightforward
computation.
Let us assume now that the statement holds for $\SRN-1$ and let
us show it for $\SRN$. The average values ${\tt A}_{\SRN}(\Lambda )$ and ${\tt D}_{\SRN}(\Lambda )$ satisfy by definition the factorization:
\begin{equation}
{\tt A}_{\SRN} (\Lambda )={\tt A}_{1}^{(\SRN)}(\Lambda ){\tt A}_{\SRN-1}^{(\SRN-1,...,1)}(\Lambda ),\text{ \ }{\tt D}_{\SRN} (\Lambda
)={\tt D}_{1}^{(\SRN)}(\Lambda ){\tt D}_{\SRN-1}^{(\SRN-1,...,1)}(\Lambda ),
\end{equation}
where the upper indices are referred to the quantum sites involved while the
lower indices to the total number of sites. We can use now the induction
hypothesis to get the result:
\begin{align}
{\tt A}_{\SRN} (\Lambda )& =(\mathcal{A}_{1}^{(\SRN)}-%
\mathcal{B}_{1}^{(\SRN)}(\Lambda ))(\mathcal{A}_{\SRN-1}^{(\SRN-1,...,1)}(\Lambda )-%
\mathcal{B}_{\SRN-1}^{(\SRN-1,...,1)}(\Lambda ))=\mathcal{A}_{\SRN}(\Lambda )-%
\mathcal{B}_{\SRN}(\Lambda ), \\
{\tt D}_{\SRN} (\Lambda )& =(\mathcal{A}_{1}^{(\SRN)}+%
\mathcal{B}_{1}^{(\SRN)}(\Lambda ))(\mathcal{A}_{\SRN-1}^{(\SRN-1,...,1)}(\Lambda )+%
\mathcal{B}_{\SRN-1}^{(\SRN-1,...,1)}(\Lambda ))=\mathcal{A}_{\SRN}(\Lambda )+%
\mathcal{B}_{\SRN}(\Lambda ),
\end{align}
where in the last formulae we have used \rf{RRel1a} together
with the fact that $\CA_\SRN(\Lambda)=\CD_\SRN(\Lambda)$
and $\CB_\SRN(\Lambda)=\CC_\SRN(\Lambda)$ for $u_n=1$, $v_n=1$,
$n=1,\dots,\SRN$.
\end{proof}
The Lemma implies in particular
\begin{equation}\label{AADD}
{\tt A}_\SRN^{}(Z_r)\,=\,\CA_\SRN(Z_r)\,,\qquad
{\tt D}_\SRN^{}(Z_r)\,=\,\CD_\SRN(Z_r)\,,
\end{equation}
for all $r=1,\dots,\SRN$. We may therefore always find a gauge transformation \rf{gauge}
such that the coefficients
$a_\SRN^{}(\eta_r)$ and $d_\SRN^{}(\eta_r)$ in \rf{SOVBax} become equal
to
\begin{equation}\label{aadd}
a_\SRN^{}(\eta_r)\,=\,{\tt a}_\SRN^{}(\eta_r)\,,\qquad
d_\SRN^{}(\eta_r)\,=\,{\tt d}_\SRN^{}(\eta_r)\,,
\end{equation}
respectively. So from now on we will denote also the coefficients in \rf{SOVBax1} with ${\tt a}$ and ${\tt d}$ omitting the index $\SRN$ unless necessary. The ansatz \rf{Qeigenstate1} therefore indeed
yields an eigenstate of $\ST(\la)$ for each solution $Q_t(\la)$ of the functional Baxter equation \rf{BaxEV}. We are going to show that {\it all} eigenstates can be obtained in this way.

\subsection{Non-degeneracy of $\ST(\la)$-eigenvalues}\label{Compatib}

In order to analyze the equations \rf{SOVBax},
let us note that the matrix representation of the operator
$\CD_r$ defined in \rf{SOVBax} is block diagonal
with blocks labeled by $n=1,\dots,\SRN$.
Let $\Psi_{n}(\eta)\in\BC^{p}$ be the vector with components
\[
\Psi_{n,k}(\eta)\,=\,
\Psi(\eta_1,\dots,\eta_{n-1},\zeta_nq^k,\eta_{n+1},\dots,\eta_\SRN)\,.
\]
Equation \rf{SOVBax} is then equivalent to the set of linear equations
\begin{equation}\label{BAXmatrix}
D^{(r)}\cdot\Psi_r(\eta)\,=\,0\,,\qquad r=1,\dots,\SRN\,.
\end{equation}
where $D^{(r)}$ is the $p\times p$-matrix 
\begin{equation}\label{D-matrix}
\begin{pmatrix}
t(\zeta_r)   &-{\tt d}(\zeta_r)&   0        &\cdots & 0 & -{\tt a}(\zeta_r)\\
-{\tt a}(q\zeta_r)& t(q\zeta_r)&-{\tt d}(q\zeta_r)& 0     &\cdots & 0 \\
      0       & {\quad} \ddots      & &     &     &         \vdots   \\
  \vdots           &     &  \cdots    &  &       &     \vdots   \\
     \vdots         &     &   & \cdots &       &   \vdots     \\
     \vdots   &            &    &  &  \ddots{\qquad}     &   0 \\
 0&\ldots&0& -{\tt a}(q^{2l-1}\zeta_r)& t(q^{2l-1}\zeta_r) &
-{\tt d}(q^{2l-1}\zeta_r)\\
-{\tt d}(q^{2l}\zeta_r)   & 0      &\ldots      &     0  & -{\tt a}(q^{2l}\zeta_r)& t(q^{2l}\zeta_r)
\end{pmatrix}
\end{equation}
The equation \rf{BAXmatrix} can have solutions only if
${\rm det}(D^{(r)})=0$. The determinant ${\rm det}(D^{(r)})$
is a polynomial of degree $p$ in each of the $\SRN$ coefficients
of the polynomial $t(\la)$.

\begin{propn}\label{Simply1}
Given that ${\rm det}(D^{(r)})=0$,
the dimension of the space of solutions to the equation \rf{BAXmatrix}
for any $r=1,\dots,\SRN\,$ is one
for generic values of the parameters $\xi$ and $\kappa$.
\end{propn}
	
\begin{proof}
Let us decompose the $p\times p$ matrix ${D}^{(r)}$ into the
block form \begin{equation}
{D}^{(r)}\,=\,\left(\begin{matrix} v^{(r)} & E^{(r)} \\
d^{(r)} & w^{(r)}\end{matrix}\right)\,,
\end{equation}
where the
submatrix $E^{(r)}$ is a $(p-1)\times (p-1)$ matrix,
$v^{(r)}$ and $w^{(r)}$ are column and row vectors with $p-1$ components,
respectively.
We assume that ${\rm det}(D^{(r)})=0$, so existence of a solution to
${D}^{(r)}\Psi=0$ is ensured.
It is easy to see that the equation ${D}^{(r)}\Psi=0$
has a unique solution provided that ${\rm det}(E^{(r)})\neq 0$.


It remains to show that ${\rm det}(E^{(r)})\neq 0$
holds for generic values of the parameters $\xi$ and $\kappa$.
To this aim let us observe that the coefficients
${\tt a}(q^k\zeta_r^{})$ and ${\tt d}(q^k\zeta_r^{})$ appearing in
\rf{BAXmatrix} depend analytically on the parameters $\kappa$.
If ${\rm det}(E^{(r)})=0$ is not identically zero, it can therefore only
vanish at isolated points.
It therefore suffices to prove the statement in a neighborhood
of the values for the parameters $\kappa$
which are such that
\begin{equation}\label{adzero}
{\tt a}(\zeta_r^{})\,=\,0\,,\qquad {\tt d}(q^{-1}\zeta_r^{})\,=\,0\,.
\end{equation}
Such values of $\kappa$ and $\xi$ exist: Setting $\kappa_n=\pm i$ for $n=1,\dots,\SRN$, one finds that
\begin{equation}
\CB_\SRN(\Lambda)\,=\,
\prod_{n=1}^{\SRN}\left( \Lambda/X_n-X/\Lambda\right)\,,
\end{equation}
which vanishes for $\la=q^{\frac{1}{2}}\xi_n$.
We may therefore choose\footnote{%
Note that this choice implies that $v_{n}\in (-1)^{p^{\prime }/2}q^{1/2}\BS_p$.} $\zeta_n=q^{\frac{1}{2}}\xi_n$.
We then find \rf{adzero} from \rf{addef}, \rf{aadd}.

Given that \rf{adzero} holds, it is easy to see that
${\rm det}(E^{(r)})\neq 0$ . Indeed, the submatrix
$E^{(r)}_{kl}$,
is lower triangular if \rf{adzero} is valid, and it has
$-{\tt d}(q^k\zeta_r^{})$, $k=0,\dots,p-2$ as its diagonal
elements. It follows that
${\rm det}(E^{(r)})=\prod_{k=0}^{p-2}{\tt d}(q^k\zeta_r^{})$
which is always nonzero
if \rf{adzero} is satisfied.
\end{proof}

The previous results admit the following reformulation which is central
for the classification and construction of the spectrum of $\ST(\lambda )$:

\begin{thm}
For generic values of the parameters $\kappa $ and $\xi $ the spectrum of $\ST(\lambda )$ is simple and all the wave-functions $\Psi _{t}(\eta )$ can be
represented in the factorized form \rf{Qeigenstate1} with $Q_{t}$ being the
eigenvalue of the $\SQ$-operator on the eigenstate $|\,t\,\rangle $.

The eigenvectors $|\,t\,\rangle $ of $\ST(\la)$ are in one-to-one
correspondence with the polynomials $Q_{t}(\la)$ of order $2l\SRN$, with $Q_{t}(0)\neq 0$,
which satisfy the Baxter equation \rf{BaxEV} with $t(\la)$ being
an even Laurent polynomial in $\lambda $ of degree $\SRN-1$.
\end{thm}

\begin{proof}
Proposition \ref{Simply1}
implies that the spectrum of $\ST(\lambda )$
is simple. Let  $|\,t\,\rangle $ be an eigenstate of $\ST(\la)$.
Self-adjointness and mutual
commutativity of $\ST(\la)$ and $\SQ(\mu )$ imply that $|\,t\,\rangle $ is
also eigenstate of $\SQ(\la)$.
Let $Q_{t}(\la)$ be the $\SQ$-eigenvalue on $|\,t\,\rangle $.
The polynomial $Q_{t}(\la)$ is related to $t(\la)$ by the
Baxter equation \rf{BaxEV} which specialized to the values $\la=\eta _{r}$
yields the equations \rf{BAXmatrix}. It follows that
there must exist nonzero numbers $\nu_{r}$ such that
\begin{equation}\label{QvsPsi}
Q_{t}^{}(\zeta _{r}q^{k})\,=\,\nu _{r}\Psi _{r,k}(\zeta _{1},\dots ,
\zeta _{\SRN})\,.
\end{equation}
This implies that the wave-functions $\Psi (\eta )$ can be represented in the form \rf{Qeigenstate1} with
$Q_{t}$
being the eigenvalue of the $\SQ$-operator
on the eigenstate $|\,t\,\rangle $.
\end{proof}

\begin{rem} It may 
be worth noting that the equivalence with the 
Fateev-Zamolodchikov model
does not hold for odd number of lattice sites. The 
spectrum of the two models is qualitatively
different, 
being doubly degenerate in the Fateev-Zamolodchikov 
model but simple in the lattice Sine-Gordon model, as illustrated
in Appendix \ref{FZ}.
\end{rem}


\subsection{Completeness of the Bethe ansatz}

Assume we are given a solution $(\la_1,\dots,\la_{2l\SRN})$
of the Bethe equations \rf{BAE}. Let us
construct the polynomial $Q(\la)$ via 
equation \rf{Qfromzeros}. Define
\begin{equation}\label{TfromQ}
t(\la)\,:=\,\frac{{\tt a}(\la)Q(q^{-1}\la)+{\tt d}(\la)Q(q\la)}{Q(\la)}\,.
\end{equation}
$t(\la)$ is nonsingular for $\la=\la_k$, $k=1,\dots,{\rm M}$ thanks to the 
Bethe equations \rf{BAE}. The pairs $(Q(\eta_r),t(\eta_r))$ satisfy the
discrete Baxter equation by construction. Inserting this solution into
\rf{Qeigenstate1} produces an eigenstate $|\,t\,\rangle$ of 
the transfer matrix $\ST(\la)$ within the SOV-representation.

Conversely, let $|\,t\,\rangle$ be an eigenvector 
of $\ST(\la)$ with eigenvalue $t(\la)$. 
Let $Q_t'(\la)$ be the 
eigenvalue of $\SQ(\la)$ on $|\,t\,\rangle$. 
Thanks to the properties of $\SQ(\la)$ listed in 
Theorem \ref{Qprop} one may factorize $Q_t'(\la)$
in the form \rf{Qfromzeros}. The tuple of zeros 
$(\la_1',\dots,\la_{2l\SRN}')$ of  $Q_t'(\la)$ must satisfy the 
Bethe equations \rf{BAE} as follows from the Baxter
equation \rf{BAX} satisfied by $\SQ(\la)$. Inserting 
$Q_t'(\eta_r)$ into
\rf{Qeigenstate1} produces an eigenstate $|\,t'\,\rangle$ 
that must be proportional to $|\,t\,\rangle$ due 
to the simplicity of the spectrum of $\ST(\la)$.

It follows that there is a one-to-one correspondence between the solutions
to \rf{BAE} and the eigenstates of the transfer matrix
({\it Completeness of the Bethe ansatz}).


\section{The spectrum --- even number of sites}\label{ap-even}

\setcounter{equation}{0}

We will now generalize these results to 
the case of a chain with even number $N$ of sites. It turns out
that the spectrum of $\ST(\la)$ is degenerate in this case, but the
degeneracy is resolved by introducing an operator $\Theta$
which commutes both with $\ST(\la)$ and $\SQ(\la)$.
The joint spectrum of $\ST(\la)$, $\SQ(\la)$ and $\Theta$
is found to be simple.

\subsection{The $\Theta $-charge}

In the case of a lattice with $\SRN$ even quantum sites, we can introduce the
operator:
\renewcommand{\su}{{\mathsf u}}
\renewcommand{\sv}{{\mathsf v}}
\begin{equation}
\Theta =\prod_{n=1}^{\SRN}\sv_{n}^{(-1)^{1+n}}.  \label{topological-charge}
\end{equation}

\begin{propn}
\label{Lemma-Theta}$\Theta $ commutes with the transfer matrix and satisfies
the following commutation relations with the entries of the monodromy
matrix:
\begin{eqnarray}
\Theta \SC(\lambda ) &=&q\SC(\lambda )\Theta \text{, \ \ \ }[\SA(\lambda ),\Theta
]=0, \\
\SB(\lambda )\Theta &=&q\Theta \SB(\lambda ),\text{ \ \ }[\SD(\lambda ),\Theta ]=0.
\end{eqnarray}
\end{propn}

\begin{proof}
The claim
can be easily verified explicitly for $\SRN=2$. The proof for the case of general
even $\SRN=2\SRM$ follows by induction. Indeed,
\begin{equation*}
\SM_{\2\,2\SRM}\,\SM_{\1\,2(\SRN-\SRM)}\,=\,\left(\begin{matrix}
\SA_{\2\,2\SRM}\,\SA_{\1\,2(\SRN-\SRM)}+\SB_{\2\,2\SRM}\,\SC_{\1\,2(\SRN-\SRM)} &
\SA_{\2\,2\SRM}\,\SB_{\1\,2(\SRN-\SRM)}+\SB_{\2\,2\SRM}\,\SD_{\1\,2(\SRN-\SRM)} \\
\SC_{\2\,2\SRM}\,\SA_{\1\,2(\SRN-\SRM)}+\SD_{\2\,2\SRM}\,\SC_{\1\,2(\SRN-\SRM)} &
\SC_{\2\,2\SRM}\,\SB_{\1\,2(\SRN-\SRM)}+\SD_{\2\,2\SRM}\,\SD_{\1\,2(\SRN-\SRM)}
\end{matrix}
\right)\,,
\end{equation*}
which easily allows one to deduce that the claim holds if it holds for all $\SRM<\SRN$.
\end{proof}

\subsection{$T$-$\Theta $-spectrum simplicity}

\begin{lem}
Let $k\in \{-l,..,l\}$ and $|t_{k}\rangle $\ be a simultaneous eigenstate of
the transfer matrix $\ST(\lambda )$ and of the $\Theta $-charge with
eigenvalues $t_{|k|}(\lambda )$ and $q^{k}$, respectively, then $\lambda
^{\SRN}t_{|k|}(\lambda )$ is a polynomial in $\lambda ^{2}$ of degree $\SRN$ which
is a solution of the system of equations:%
\begin{equation}
\det(D^{(r)})\,=\,0\text{ \ \ }\forall r\in \{1,...,[\SRN]\}\text{%
,}  \label{system-t}
\end{equation}%
where the $p\times p$ matrices \textsc{D}$^{(r)}$\ are defined in
\rf{D-matrix}, with asymptotics of $t_{|k|}(\la)$ given  by:
\begin{equation}
\lim_{\log\lambda\rightarrow \pm\infty}\lambda ^{\mp\SRN}t_{|k|}(\lambda )=\left(
\prod_{a=1}^{\SRN}\frac{\kappa _{a}\xi _{a}^{\mp 1}}{i}\right) \left(
q^{k}+q^{-k}\right).  \label{asymptotics-t}
\end{equation}
\end{lem}

\begin{proof}

The fact that the generic eigenvalue of the transfer matrix has to satisfy
the system \rf{system-t} has been discussed in Section \ref{Spec}; so we have just to verify the
asymptotics (\ref{asymptotics-t}) for the $\ST$-eigenvalue $t_{|k|}(\lambda )$.
This follows by the assumption that $|t_{k}\rangle $ is an eigenstate of $\Theta $\ with eigenvalue $q^{k}$,
 and by formulae
\begin{equation}
\lim_{\log{\lambda} \rightarrow \pm\infty}\lambda ^{\mp\SRN}\ST(\lambda )=\left(
\prod_{a=1}^{\SRN}\frac{\kappa _{a}\xi _{a}^{\mp 1}}{i}\right) \left( \Theta
+\Theta ^{-1}\right) ,
\end{equation}
derived in appendix \ref{Asymp-A-D}.
\end{proof}

The previous Lemma implies in particular the following:

\begin{thm}
For generic values of the parameters $\kappa $ and $\xi $ the simultaneous
spectrum of $\ST$ and $\Theta $\ operators is simple and the generic
eigenstate $|t_{k}\rangle$ of the $\ST$-$\Theta $-eigenbasis has a wave-function of the form
\begin{equation}\label{Psi-factor}
\Psi(\eta )=\eta _{\SRN}^{-k}\prod_{a=1}^{\SRN-1}\psi_{|k|}(\eta
_{a}),
\end{equation}
where, for any $r\in \{1,...,\SRN-1\}$, the vector $(\psi_{|k|}(\zeta
_{r}),\psi_{|k|}(\zeta _{r}q),...,\psi_{|k|}(\zeta _{r}q^{2l}))$ is the unique (up
to normalization) solution of the linear equations (\ref{BAXmatrix})
corresponding to $t_{|k|}(\lambda )$.
\end{thm}

\begin{proof}

Let us use the SOV-construction of $\ST$-eigenstates and let us observe that
an analog of Proposition \ref{Simply1} also holds\footnote{%
The proof given previously holds for both the cases $\SRN$ even and odd just
changing $\SRN$ into $[\SRN]$ everywhere.} for even $\SRN$. This implies
that the wave-function $\Psi(\eta )$ can be represented in the form
\begin{equation}
\Psi(\eta )=f_{t_{k}}(\eta _{\SRN})\prod_{a=1}^{\SRN-1}\psi_{|k|}(\eta
_{a})\,.
\end{equation}
Finally, using that $|t_{k}\rangle $ is eigenstate of $\Theta $\ with eigenvalue $q^{k}$ we get
$f_{t_{k}}(\eta _{\SRN})\propto \eta _{\SRN}^{-k}$.
\end{proof}

Thanks to the explicit construction of the simultaneous $\ST$-$\Theta $\
eigenstates given in \rf{Psi-factor}, we have that the eigenstates of $\ST(\lambda )$
with $\Theta $-charge eigenvalue 1 are simple, while all the others are
doubly degenerate with eigenspaces generated by a pair of
$\ST$-eigenstates with $\Theta $-charge eigenvalues $q^{\pm k}$.

\subsection{$\SQ$-operator and Bethe ansatz}

Let us point out some peculiarity of the $\SQ$-operator
in the case of even chain. In order to see this, we need the following Lemma which is of interest in its own right.

\begin{lem}\label{Unilem}
For a given $t(\la)$, there is at most one polynomial of degree $2l\SRN$ which satisfies the
Baxter equation \rf{BaxEV}.
\end{lem}
\begin{proof}

Let us define the q-Wronskian:
\begin{equation}\label{q-W}
W(\la)\,=\,Q_1(\la)Q_2(q^{-1}\la)-Q_2(\la)Q_1(q^{-1}\la)\,.
\end{equation}
written in terms of two solutions $Q_{1}(\lambda )$ and $Q_{2}(\lambda )$ of
the Baxter equation; then $W(\lambda)$ satisfies the equation
\begin{equation}
{\tt a}(\lambda )\,W(\lambda )\,=\,{\tt d}(\lambda )\,\ST^{+}W(\lambda )\,.
\end{equation}%
Note now that Lemma \ref{A=A-B} implies:%
\begin{equation}
\prod_{k=0}^{2l}{\tt a}(\lambda q^{k})\,\neq \,\prod_{k=0}^{2l}{\tt d}(\lambda q^{k}),%
\text{ \ \ \ }\forall \lambda \notin \mathbb{B}_{\SRN},  \label{degcond}
\end{equation}%
so for any $\lambda \notin \mathbb{B}_{\SRN}$ the only solution consistent with cyclicity $(\ST^{+})^{p}=1$ is $W(\lambda )\equiv 0$.
It is then easy to see that this implies that $Q_1(\la)=Q_2(\la)$.\end{proof}
Now we can prove the following: 

\begin{propn}
The $\SQ$-operators commute with the $\Theta $-charge and $|t_{\pm
|k|}\rangle $ are $\SQ$-eigenstates with common eigenvalue $Q_{|k|}(\lambda )$
of degree $2l\SRN-k(a_{\infty }^{\pm }p\pm 1)$ in $\lambda $ and a zero of
order $k(a_{0}^{\pm }p\pm 1)$ at $\lambda =0$, where $a_{0}^{+}$
and $a_{\infty }^{+}$ are non-negative integers, while $a_{0}^{-}$ and
$a_{\infty }^{-}$ are positive integers.
\end{propn}

\begin{proof}
The commutativity of $\ST$ and $\SQ$-operators implies that the $\ST$-eigenspace $%
\mathcal{L}(|t_{\pm |k|}\rangle )$\ corresponding to the eigenvalue $%
t_{|k|}(\lambda )$ is invariant under the action of $\SQ$ and so for $k=0$ any
$T$-eigenstate $|t_{0}\rangle $ is directly a $\SQ$-eigenstate. Let us observe
that the self-adjointness of $\SQ$ implies that in the two-dimensional $\ST$-eigenspace
$\mathcal{L}(|t_{\pm |k|}\rangle )$ with $k\neq 0$ we can always take two
linear combinations of the states $|t_{|k|}\rangle $ and $|t_{-|k|}\rangle $
which are $\SQ$-eigenstates. Now thanks to the Lemma \ref{Unilem} for  fixed $\ST$%
-eigenvalue $t_{|k|}(\lambda )$ the corresponding $\SQ$-eigenvalue $Q_{|k|}(\lambda )$\ is
unique which implies that $|t_{\pm|k|}\rangle $ are themselves $\SQ$-eigenstates. The commutativity of the $\SQ$%
-operator with the $\Theta $-charge follows by observing that the $|t_{\pm
|k|}\rangle $ define a basis.

Let us complete the proof showing that the conditions on the polynomial $%
Q_{|k|}(\lambda )$ stated in the Proposition are simple consequences of the fact that $|t_{\pm
|k|}\rangle $ are eigenstates of the $\Theta $-charge with eigenvalues $%
q^{\pm |k|}$. Indeed, the compatibility of the asymptotics conditions (\ref%
{asymptotics-t}) with the $\ST\SQ$ Baxter equation implies
\begin{equation}\label{Asymp-Q_k}
\lim_{\lambda \rightarrow 0}\frac{Q_{|k|}(\lambda q)}{Q_{|k|}(\lambda )}%
=q^{\pm |k|},\text{ \ \ }\lim_{\lambda \rightarrow \infty }\frac{%
Q_{|k|}(\lambda q)}{Q_{|k|}(\lambda )}=q^{-(\SRN\pm |k|)},
\end{equation}%
which are equivalent to the conditions on
the polynomial $Q_{|k|}(\lambda )$ stated in the Proposition.
\end{proof}

Note that the uniqueness of the $\SQ$-eigenvalue $Q_{|k|}(\lambda )$ corresponding to a given $\ST$-eigenvalue $t_{|k|}(\lambda )$ implies that each vector $(\psi_{|k|}(\zeta_{r}),\psi_{|k|}(\zeta _{r}q),...,\psi_{|k|}(\zeta _{r}q^{2l}))$ appearing in \rf{Psi-factor} must be proportional to the vector $(Q_{|k|}(\zeta
_{r}),Q_{|k|}(\zeta _{r}q),...,Q_{|k|}(\zeta _{r}q^{2l}))$ so that the previous results admit the following reformulation:
\begin{thm}
The pairs of eigenvectors $|t_{|k|}\rangle $ and $|t_{-|k|}\rangle $ of $%
\ST(\la)$ are in one-to-one correspondence with the polynomials $Q_{|k|}(\la)
$ of maximal order $2l\SRN$ which have the asymptotics (\ref{Asymp-Q_k}) and
satisfy the Baxter equation \rf{BaxEV} with $t_{|k|}(\la)$ being an even Laurent polynomial
in $\lambda $ of degree $\SRN$.
\end{thm}

As in the case of $\SRN$ odd this reformulation allows the classification and
construction of the spectrum of $\ST(\lambda )$ by the analysis of the
solutions to the system of the Bethe equations.

\appendix

\section{Cyclic solutions of the star-triangle relation}

\setcounter{equation}{0}

It will sometimes be convenient for us to identify $\BZ_p\equiv \BZ/p\BZ$ with the
subset $\BS_p=\{q^{2n};n=-l,\dots,l\}$ of the unit circle since $q^{2l+1}=1$.


\subsection{Definition and elementary properties}

\subsubsection{The function $w_\la(z)$}
Let us define a function $w_\la:\BS_p \ra\BC$ by
\begin{align}
& w_{\la}(q^{2n})\,=\,\prod_{r=1}^{n}\frac{1+\la q^{2r-1}}{\la+
q^{2r-1}}
\prod_{r=1}^{l}\frac{\la+ q^{2r-1}}{1+q^{2r-1}}\,,\qquad n=0,\dots,p-1\,.
\end{align}
This function is indeed cyclic (defined on $\BS_p$) since $\prod_{k=1}^p(1-x q^{2k})=1-x^p$ implies
\begin{equation}
w_{\la}(q^{2p})\,=\,w_{\la}(q^{4l+2})\,=\,w_{\la}(1)\,,
\end{equation}
The
function $w_\la(z)$ is the unique  solution to the functional
equation
\begin{align}\label{fun1}
& (z+\la)w_\la(qz)\,=\,(1+\la z)w_\la(q^{-1}z)\,,
\end{align}
which
is a  polynomial of order $l$ in $\la$ and which satisfies the
normalization condition
\begin{equation}\label{normcond}
w_1(q^{n})\,=\,1\qquad \forall\; n\in\BZ_p\,.
\end{equation}
The function $w_\la(z)$ satisfies the inversion relation
\begin{equation}\label{invrel}
w_{\la}(z)w_{1/\la}(z)\,=\,\chi_\la^{}\,,\qquad\chi_\la^{}\,=\,\la^{-l}\prod_{r=1}^{l}
\frac{(\la+q^{2r-1})(1+\la q^{2r-1})}{(q^{2r-1}+1)^2}\,.
\end{equation}

\subsubsection{The function $\overline{w}_{\la}(z)$}
Let us also introduce the function $\overline{w}_{\la}(z)$ as the
discrete Fourier transformation of $w_\la$,
\begin{align}
& \overline{w}_{\la}(z)\,=\,\frac{1}{p}\sum_{k=-l}^l
z^k\,w_{\la}(q^{k})
\end{align}
$\overline{w}_\la(z)$ can be characterized as the unique
solution to the functional relation
\begin{align}\label{fun2}
& (1-\la qz)\overline{w}_\la(qz)\,=\,(z-q\la)
\overline{w}_\la(q^{-1}z)\,,
\end{align}
which is a polynomial of order $l$ in $\la$ and which satisfies
the normalization condition $\overline{w}_1(q^{n})=\de_{n,0}$.
It may therefore be represented by the product
\begin{equation}
\overline{w}_\la(q^{2n})\,=\,\prod_{r=1}^{n}\frac{q\la-q^{2r-1}}{\la q^{2r}-1}\prod_{s=1}^l\frac{\la q^{2s}-1}{q^{2s}-1}\,.
\end{equation}
It is also useful to observe that $\overline{w}_\la$ and $w_\la$ are related by complex conjugation as follows:
\begin{equation}\label{w-wbar}
(w_{\ep\la}(z))^\ast\,=\,\overline{w}_{\ep\la^\ast}(z)\, \prod_{s=1}^l\frac{1-q^{2s}}{1+q^{2s-1}}\,.
\end{equation}
This relation makes it easy to deduce properties of $\overline{w}_\la$ from those of $w_\la$.

\subsubsection{Further functional relations}

Let us list further functional relations satisfied by the function $w_\la(z)$.
\begin{equation}\label{mixedfunrel}\begin{aligned}
& (\la+z)\,w_\la^{}(qz)\,=\,q^{l^2+l}\,{z}^{+\frac{1}{2}}\,(1+q\la)\,w_{q\la}^{}(z),\\
& (1+\la)\,w_\la^{}(qz)\,=\,q^{l^2+l}\,{z}^{-\frac{1}{2}}\,(1+\la z)\,w_{\la/q}^{}(z),\\
& (1-q\la)\,\overline{w}_\la^{}(qz)\,=\,q^{-l^2-l}\,{z}^{-\frac{1}{2}}\,(z-q\la)\,\overline{w}_{q\la}^{}(z),\\
& (1-q\la z)\,\overline{w}_\la^{}(qz)\,=\,q^{-l^2-l}\,{z}^{+\frac{1}{2}}\,(1-\la)\,\overline{w}_{\la/q}^{}(z).
\end{aligned}
\end{equation}
These relations play a key role in the derivation of the Baxter equation \rf{BAX}.

\subsection{Star-triangle relation}

One of the most important properties of the function $w_\la(x)$ is the star-triangle relation \cite{FZ}
\begin{align}\label{STR}
\sum_{x\in\BS_p}\;\overline{w}_{\al}(x/u)\,w_{\al\be}(x/v)\,\overline{w}_{\be}(x/w)\,=\,
w_{\al}(w/v)\,\overline{w}_{\al\be}(u/w)\,w_{\be}(v/u)\,,
\end{align}
see \cite{Ba08} for an elegant proof and references to related work.
We are mainly going to use the following consequence of \rf{STR} called the exchange relation
\begin{align}\label{Ex}
\sum_{y\in\BS_p}\;\overline{w}_{\al}^{}(y/u)\,& w_{\be}^{}(y/v)\,\overline{w}_{\ga}^{}(y/w)\,w_\de^{}(y/x)\,=\,\\
&=\,\frac{w_{\be/\al}^{}(u/v)}{w_{\be/\al}^{}(x/w)}\sum_{y\in\BS_p}\;\overline{w}_\be^{}(y/u)\,w_{\al}^{}(y/v)\,\overline{w}_\de^{}(y/w)w_\ga^{}(y/x)\,.\nonumber
\end{align}
for $\al\ga/\be\de=1$.
In order to prove \rf{Ex} let us note
the relation
\begin{equation}\label{Inv}
\sum_{z\in\BS_p}\;\overline{w}_\al(u/z)\overline{w}_{1/\al}(z/v)\,=\,\frac{1}{p}\sum_{k=-l}^l(u/v)^k w_\al(q^k)w_{1/\al}(q^k)\,=\,\de_{u,v}\chi_\al\,,
\end{equation}
since $ \chi_\al\equiv w_\al(z)w_{1/\al}(z)$ is independent of $z$.
By inserting \rf{Inv} into the left hand side of \rf{Ex} we may therefore calculate
\begin{align*}
\sum_{y\in\BS_p} & \; \overline{w}_{\al}^{}(u/y)\, w_{\be}^{}(y/v)\,\overline{w}_{\ga}^{}(y/w)\,w_\de^{}(x/y)\,=\,\\
& =\,\chi_\al^{-1}\sum_{y\in \BS_p} \sum_{z\in\BS_p}\sum_{y'\in\BS_p}\; \overline{w}_{\al}^{}(y/u)\, w_{\be}^{}(y/v)\,\overline{w}_{\be/\al}^{}(y/z)
\;\overline{w}_{\de/\ga}^{}(z/y')\,w_\de^{}(y'/x)\,\overline{w}_{\ga}^{}(y'/w)\\
& =\,\chi_\al^{-1}\sum_{z\in\BS_p} \; {w}_{\al}^{}(v/z)\, \overline{w}_{\be}^{}(z/u)\,{w}_{\be/\al}^{}(u/v)
\;{w}_{\de/\ga}^{}(w/x)\,\overline{w}_\de^{}(z/w)\,{w}_{\ga}^{}(x/z)\,.\end{align*}
The sums over $y$ and $y'$ have been carried out with the help of the star-triangle relation \rf{STR}.
It remains to recall that $\chi_\al^{-1}{w}_{\de/\ga}^{}(w/x)=({w}_{\be/\al}^{}(w/x))^{-1}$ to
complete the proof of \rf{Ex}.

\section{Properties of the $\SQ$-operator}\label{Qproofs}

\setcounter{equation}{0}

\subsection{Proof of the Baxter equation}

The strategy is similar to  \cite{Ba72,BS}. Consider
\renewcommand{\BT}{{\mathbf T}}
\newcommand{\BL}{{\mathbf L}}
\begin{equation}
\BT(\la)\cdot\langle\,{\mathbf z}\,|\,\SY(\la)\,|\,{\mathbf z}'\,\rangle\,\equiv\,
\langle\,{\mathbf z}\,|\,\ST(\la)\,\SY(\la)\,|\,{\mathbf z}'\,\rangle\,.
\end{equation}
The operator ${\mathbf T}(\la)$ is the  difference operator obtained by
replacing $L_n^{\rm\sst SG}(\la)\ra \BL_n^{\rm\sst SG}(\la)$ in \rf{Mdef}, with ${\mathbf L}_n^{\rm\sst SG}(\la)$
obtained from \rf{Lmod} by replacing $\su_n$ and $\sv_n$ by the corresponding multiplication and
shift operators ${\mathbf u}_n$ and ${\mathbf v}_n$ defined in \rf{reprdef},
\renewcommand{\su}{{\mathbf u}}
\renewcommand{\sv}{{\mathbf v}}
\begin{equation}\label{Lmod}
\BL_n^{\rm\sst SG}\,= \,\frac{\kappa_n}{i}\, \left( \begin{array}{cc}
-\su_n^{}(\vartheta^{-1}_n\sv_n^{}-\vartheta_n^{}\sv_n^{-1}) &
\la_n^{}\sv_n^{} - \la^{-1}_n \sv_n^{-1}  \\
 \la_n^{} \sv_n^{-1} - \la^{-1}_n \sv_n^{} &
\su_n^{-1}(\vartheta_n^{}\sv_n^{}-\vartheta^{-1}_n\sv_n^{-1})
 \end{array} \right)\,.
\end{equation}
In writing \rf{Lmod} we have introduced the short-hand notation
$\vartheta_n=iq^{\frac{1}{2}}\kappa^{-1}_n$and $\la_r\equiv\la/\xi_n$.
Note that ${\mathbf T}(\la)$
acts on the argument $\bz=(z_1,\dots,z_\SRN)$ of $Y_\la(\bz,\bz')$, while it does not act on $\bz'$.
In order to simplify the expression for $\BT(\la)$ we may therefore use a gauge-transformation of the form
\begin{equation}
\tilde{\BL}^{\rm\sst SG}_n(\la)\,=\,g_{n+1}^{}\,\BL^{SG}_n(\la)\,g_n^{-1}\,,\qquad g_n=\left(\begin{matrix} 1 & 0 \\ z_n' & 1
\end{matrix}\right)\,.
\end{equation}
The key point to observe is that
\begin{equation}\label{Btozero}\begin{aligned}
\frac{i}{\kappa_n}\tilde{\BL}^{\rm\sst SG}_n(\la)_{21}^{}\cdot Y_\la^{}(\bz,\bz')
\,& =\,\la^{-1}_n\, (z_n'+\la_n^{}\vartheta_n^{}\su_n^{})\,(z_{n+1}'+\la_n^{}\vartheta_n^{-1}\su_n^{-1})\sv_n^{-1}\cdot Y_\la^{}(\bz,\bz')\\
&\quad-\la_n^{-1}\,(1+\la_n^{}\vartheta_n^{}z_n'\su_n^{-1})\,(1+\la_n^{}\vartheta_n^{-1}z_{n+1}'\su_n^{})\sv_n^{}\cdot Y_\la^{}(\bz,\bz')\\
&=\,0\,,
\end{aligned}
\end{equation}
the last step being an easy consequence of the recursion relations \rf{fun1},
\rf{fun2} satisfied by the functions $w_\la(z)$ and $\overline{w}_\la(z)$
which appear in the kernel $Y_\la(\bz,\bz')$.

Equation \rf{Btozero} implies that
\begin{equation}
\BT(\la)\cdot Y_\la(\bz,\bz')\,=\,
\Bigg(\,\prod_{n=1}^{\SRN}\tilde{\BL}^{\rm\sst SG}_n(\la)_{11}^{}+
\prod_{n=1}^{\SRN}\tilde{\BL}^{\rm\sst SG}_n(\la)_{22}^{}\Bigg)\cdot
Y_\la(\bz,\bz')\,.
\end{equation}
We have
\begin{align*}
\tilde{\BL}^{\rm\sst SG}_n(\la)_{11}^{}\cdot Y_\la^{}(\bz,\bz')\,=&\,-\frac{\kappa_n}{i}z'_n\big[\vartheta_n^{-1}(z_n^{}/z_n'+\la_n^{}\vartheta_n^{})\sv_n^{}-\la_n^{-1}(1+\vartheta_n^{}\la_n^{} z_n^{}/z_n')\sv_n^{-1}\big]\cdot Y_\la^{}(\bz,\bz')\\
=&-\frac{\kappa_n}{i}z_n'(\la_n^{}/\vartheta_n^2-1/\la_n^{})\frac{1+\la_n^{}\vartheta_n^{} z_n^{}/z_n'}{1+z_n^{}z_{n+1}'\la_n^{}/\vartheta_n^{}}\,\sv_n^{-1}\cdot Y_\la^{}(\bz,\bz')\,.
\end{align*}
By using the recursion relations \rf{mixedfunrel} one may rewrite this as
\begin{align*}
\tilde{\BL}^{\rm\sst SG}_n(\la)_{11}^{}\cdot Y_\la^{}(\bz,\bz')\,=&\,\frac{\kappa_n}{i}(z_n'/z_{n+1}')^{\frac{1}{2}}\,(1/\la_n-1/\vartheta_n)(1+q^{-1}\la_n\vartheta_n)\,Y_{q^{-1}\la}^{(n)}\,\prod_{r\neq n}Y_\la^{(r)}.
\end{align*}
where $Y_\la^{(n)}\,\equiv\,\overline{w}_{\ep\la/\kappa_n\xi_n}^{}(z_n^{}/z_n')\,w_{\ep\la\kappa_n/\xi_n}^{}(z_n^{}z_{n+1}')$.
We may similarly calculate
\begin{align*}
\tilde{\BL}^{\rm\sst SG}_n(\la_n)_{22}^{}\cdot Y_\la^{}(\bz,\bz')\,=&\,\frac{\kappa_n}{i}\big[(\vartheta_n^{} z_n^{-1}+z_{n+1}^{}\la_n^{})\sv_n^{}-(\vartheta_n^{-1} z_n^{-1}+z_{n+1}^{}\la_n^{-1})\sv_n^{-1}\big]\cdot Y_\la^{}(\bz,\bz')\\
=&-\frac{\kappa_n}{i}(z_n')^{-1}(\la_n-1/\la_n^{}\vartheta_n^2)\frac{1+z_{n+1}'z_n\vartheta_n^{}/\la_n^{}}{1+z_n^{}/z_n'\la_n^{}\vartheta_n^{}}\,\sv_n^{-1}\cdot Y_\la^{}(\bz,\bz')\\
=&-\frac{\kappa_n}{i}(z_{n+1}'/z_n')^{\frac{1}{2}}\,(1/\la_n^{}+q/\vartheta_n^{})(1-\la_n^{}\vartheta_n^{})\,Y_{q\la}^{(n)}\,\prod_{r\neq n}Y_\la^{(r)}.
\end{align*}
It follows that
\begin{align*}
\prod_{n=1}^{\SRN}\tilde{\BL}^{\rm\sst SG}_n(\la)_{11}^{}\cdot Y_\la^{}(\bz,\bz')\,=\,&\prod_{n=1}^{\SRN}\frac{\kappa_n}{i}(1/\la_n-1/\vartheta_n)(1+q^{-1}\la_n\vartheta_n)\cdot Y_{q^{-1}\la}^{}(\bz,\bz')\,,\\
\prod_{n=1}^{\SRN}\tilde{\BL}^{\rm\sst SG}_n(\la)_{22}^{}\cdot Y_\la^{}(\bz,\bz')\,=\,&\prod_{n=1}^{\SRN}i\kappa_n(1/\la_n^{}+q/\vartheta_n^{})(1-\la_n^{}\vartheta_n^{})\cdot Y_{q\la}^{}(\bz,\bz')\,.
\end{align*}
This concludes the proof.

\subsection{Proof of the commutativity}

The key observation to be made is the fact that the operators $\SY(\la)$ satisfy the exchange relation
\begin{equation}\label{Yex}
\SY(\la)\cdot(\SY(\mu^\ast))^{\dagger}\,=\,\SY(\mu)\cdot(\SY(\la^\ast))^{\dagger}\,.
\end{equation}
This is an easy consequence of the exchange relation \rf{Ex}, as
observed in \cite{BS}.
Since we have $\la_n/\mu_n=\la_m/\mu_m$ for all $n,m=1,\dots,\SRN$
we may calculate
\begin{align*}
\langle\,&{\mathbf z}\,|\,\SY(\la)\, (\SY(\mu^\ast))^{\dagger}\,|\,{\mathbf z}'\,\rangle\,=\,\\
&=\sum_{{\mathbf y}\in\BS_p^{\SRN}}\;
\prod_{n=1}^{\SRN}\overline{w}^{}_{\ep\la_n/\kappa_n}(z_n^{}/y_n^{})w^{}_{\ep\la_n\kappa_n}(z_n^{}y_{n+1}^{})
\overline{\overline{w}^{}_{\ep\mu_n^\ast/\kappa_n^{}}(y_n^{}/z_n')w^{}_{\ep\mu_n^\ast\kappa_n^{}}(y_{n+1}^{}z_n')}\\
&=\phi_{0}\sum_{{\mathbf y}\in\BS_p^{\SRN}}\;
\prod_{n=1}^{\SRN}\overline{w}^{}_{\ep\la_n/\kappa_n}(z_n^{}/y_n^{})
{w}^{}_{\ep\mu_n/\kappa_n}(y_n^{}/z_n')\overline{w}^{}_{\ep\mu_{n-1}\kappa_{n-1}}(y_{n}^{}z_{n-1}')w^{}_{\ep\la_{n-1}\kappa_{n-1}}(z_{n-1}^{}y_{n}^{})\\
&=\phi_{0}\sum_{{\mathbf y}\in\BS_p^{\SRN}}\;
\prod_{n=1}^{\SRN}\overline{w}^{}_{\ep\mu_n/\kappa_n}(z_n^{}/y_n^{})
{w}^{}_{\ep\la_n/\kappa_n}(y_n^{}/z_n')\overline{w}^{}_{\ep\la_{n-1}\kappa_{n-1}}(y_{n}^{}z_{n-1}')w^{}_{\ep\mu_{n-1}\kappa_{n-1}}(z_{n-1}^{}y_{n}^{})\\
&=\sum_{{\mathbf y}\in\BS_p^{\SRN}}\;
\prod_{n=1}^{\SRN}\overline{w}^{}_{\ep\mu_n/\kappa_n}(z_n^{}/y_n^{})w^{}_{\ep\mu_n\kappa_n}(z_n^{}y_{n+1}^{})
\overline{\overline{w}^{}_{\ep\la_n^\ast/\kappa_n^{}}(y_n^{}/z_n')w^{}_{\ep\la_n^\ast\kappa_n^{}}(y_{n+1}^{}z_n')}\\
&=\,\langle\,{\mathbf z}\,|\,\SY(\mu)\, (\SY(\la^\ast))^{\dagger}\,|\,{\mathbf z}'\,\rangle\,,
\end{align*}
where $\phi_{0}\equiv(-1)^{l\SRN}q^{2\SRN l(l+1)}$ as it follows by formula \rf{w-wbar}.
The mutual commutativity of the operators $\SQ$ is an easy consequence.
Let us furthermore note that $({\tt a}(\la))^\ast={\tt d}(\la)$ and $(\ST(\la))^\dagger=\ST(\la)$ for $\la\in\BR$. Using \rf{Yex} we may calculate
\begin{align*}
\ST(\la)\cdot\SY(\la)\cdot(\SY(\mu))^{\dagger}\,&=\,\big[{\tt a}(\la)\SY(q^{-1}\la)+{\tt d}(\la)\SY(q\la)\big]\cdot(\SY(\mu))^\dagger\\
&=\,{\tt a}(\la)\,\SY(\mu)\cdot(\SY(q\la))^\dagger+{\tt d}(\la)\SY(\mu)\cdot(\SY(q^{-1}\la))^\dagger\\
&=\,\SY(\mu)\cdot\big[{\tt d}(\la)\SY(q\la)+{\tt a}(\la)\SY(q^{-1}\la)\big]^\dagger\\
&=\,\SY(\mu)\cdot\big[\ST(\la)\cdot\SY(\la)\big]^{\dagger}\\
&=\,\SY(\mu)\cdot(\SY(\la))^\dagger\cdot\ST(\la),\\
\end{align*}
which obviously implies $[\ST(\la),\SQ(\la)]=0$.

\subsection{Proof of integrability}

In order to prove \rf{UfromQ} first note that \rf{FT} allows us to write
\begin{equation}
\langle\,z_n^{}\,|\,w_{\la}(\sf_{2n})\,|\,z_n'\,\rangle\,=\,\sum_{r=-l}^l \,\langle\,z_n^{}\,|\,\sf_{2n}^{-r}\,|\,z_n'\,\rangle\,\overline{w}_\la(q^r)\,.
\end{equation}
Noting that $\langle\,q^{2k_n}\,|\,\sf_{2n}^{-r}\,|\,q^{2k_n'}\,\rangle=\langle\,q^{2k_n-2r}\,|\,q^{2k_n'}\,\rangle\,=\,\de_{r,k_n'-k_n^{}}$
we find that
\begin{equation}
\langle\,z_n^{}\,|\,w_{\la}(\sf_{2n})\,|\,z_n'\,\rangle\,=\,\overline{w}_\la(z_n^{}/z_n')\,.
\end{equation}
Thanks to this identity and \rf{normcond} it is easy to see that
\begin{align*}
\langle\,\bz\,|\,\SY(1/\kappa\ep)\,|\,\bz'\,\rangle\,=\,&\prod_{n=1}^{\SRN}\de_{z_n',z_n}\overline{w}_{\kappa^{-2}}(z_n^{}/z_n')\,=\,
\langle\,\bz\,|\,\prod_{n=1}^{\SRN}w_{\kappa^{-2}}(\sf_{2n})\,|\,\bz'\,\rangle\,,
\end{align*}
which implies
\begin{align}
\SQ^+(1/\kappa\ep)\,=\,\prod_{n=1}^{\SRN}w_{\kappa^{-2}}(\sf_{2n})\cdot\SY_\infty^{\dagger}\,.
\end{align}
Similarly note that
\begin{align*}
\langle\,\bz\,|\,\SY(\kappa/\ep)\,|\,\bz'\,\rangle\,=\,&\prod_{n=1}^{\SRN}{w}_{\kappa^{2}}(z_n^{}z_{n+1}')\,=\,
\langle\,\bz\,|\,\prod_{n=1}^{\SRN}w_{\kappa^{2}}(\sf_{2n+1})\,|\,\bz'\,\rangle\,,
\end{align*}
which implies
\begin{align}
\SQ^-(\kappa/\ep)\,=\,\SY_0\cdot\prod_{n=1}^{\SRN}(w_{\kappa^2}(\sf_{2n+1}))^{-1}\,.
\end{align}
It remains to notice that $\SY_\infty^{\dagger}\cdot\SY_0=\SU_0$
to conclude the proof. Indeed, using the notation
$\bz=(q^{2k_1},\dots,q^{2k_\SRN})$ and
$\bz''=(q^{2k_1''},\dots,q^{2k_\SRN''})$, we may calculate
\[\begin{aligned}
\langle\,\bz\,|\,\SY_\infty^{\dagger}\cdot\SY_0\,|\,\bz''\rangle\,&=\,
\frac{1}{p^{\SRN}}\sum_{(k_1',\dots,k_\SRN')\in\BZ_p^\SRN}\,\prod_{n=1}^{\SRN}
q^{-2k_n'(k_n+k_{n+1})}q^{-2k_n'(k_n''+k_{n+1}'')}\\
&=\,\prod_{n=1}^{\SRN}\de_{k_n+k_{n+1}+k_n''+k_{n+1}'',0}^{}\,=\,
\prod_{n=1}^{\SRN}\de_{k_n,-k_n''}^{}\\
&=\,\langle\,\bz\,|\,\SU_0\,|\,\bz''\rangle\,, \\
\end{aligned}
\]
keeping in mind that we consider the case of odd $\SRN$.

\section{Asymptotics of Yang-Baxter generators}\label{Asymp-A-D}

\setcounter{equation}{0}

\renewcommand{\su}{{\mathsf u}}
\renewcommand{\sv}{{\mathsf v}}

From the known form of the Lax operator we derive the following asymptotics
for $\lambda \rightarrow +\infty $ and $0$ of the generators of the
Yang-Baxter algebras.

\begin{itemize}
\item[] \textbf{\SRN \ odd:} The leading operators are $\SB_{\SRN}(\lambda )$
and $\SC_{\SRN}(\lambda )$\ with asymptotics:
\begin{align}
\SB_{\SRN}(\lambda )& =\left( \prod_{a=1}^{\SRN}\frac{\kappa _{a}}{i}\right) \left(
\lambda ^{\SRN}\prod_{a=1}^{\SRN}\frac{\sv_{a}^{(-1)^{1+a}}}{\xi _{a}}-\lambda
^{-\SRN}\prod_{a=1a}^{\SRN}\xi _{a}\sv_{a}^{(-1)^{a}}\right) +\text{sub-leading terms}, \\
\SC_{\SRN}(\lambda )& =\left( \prod_{a=1}^{\SRN}\frac{\kappa _{a}}{i}\right) \left(
\lambda ^{\SRN}\prod_{a=1}^{\SRN}\frac{\sv_{a}^{(-1)^{a}}}{\xi _{a}}-\lambda
^{-\SRN}\prod_{a=1}^{\SRN}\xi _{a}\sv_{a}^{(-1)^{1+a}}\right) +\text{sub-leading terms}.
\end{align}

\item[] \textbf{\SRN \ even:} The leading operators are $\SA_{\SRN}(\lambda )$
and $\SD_{\SRN}(\lambda )$\ with asymptotics:
\begin{align}
\SA_{\SRN}(\lambda )& =\left( \prod_{a=1}^{\SRN}\frac{\kappa _{a}}{i}\right) \left(
\lambda ^{\SRN}\prod_{a=1}^{\SRN}\frac{\sv_{a}^{(-1)^{1+a}}}{\xi _{a}}+\lambda
^{-\SRN}\prod_{a=1}^{\SRN}\xi _{a}\sv_{a}^{(-1)^{a}}\right) +\text{sub-leading terms},
\label{asymp-A} \\
\SD_{\SRN}(\lambda )& =\left( \prod_{a=1}^{\SRN}\frac{\kappa _{a}}{i}\right) \left(
\lambda ^{\SRN}\prod_{a=1}^{\SRN}\frac{\sv_{a}^{(-1)^{a}}}{\xi _{a}}+\lambda
^{-\SRN}\prod_{a=1}^{\SRN}\xi _{a}\sv_{a}^{(-1)^{1+a}}\right) +\text{sub-leading terms}.
\end{align}
\end{itemize}

Note that these asymptotics imply for the SOV-representation of the
Yang-Baxter generators the following formulae\footnote{%
Note that the transformation $\mathsf{W}^{\text{SOV}}$ is meant to act as a
similarity transformation in the space of the representation, i.e. $\mathsf{W%
}^{\text{SOV}}\equiv \mathsf{w}^{\text{SOV}}I$ where $\mathsf{w}^{\text{SOV}%
} $ is a non-trivial operator on space of the states.}:

\begin{itemize}
\item[] \textbf{\SRN \ odd:}
\begin{equation}
\left( \mathsf{w}^{\text{SOV}}\right) ^{-1}\left(
\prod_{a=1}^{\SRN}\sv_{a}^{(-1)^{1+a}}\right) \mathsf{w}^{\text{SOV}%
}=\prod_{a=1}^{\SRN}\frac{\xi _{a}}{\eta _{a}}.
\end{equation}

\item[] \textbf{\SRN \ even:}
\begin{eqnarray}
\prod_{a=1}^{\SRN}\xi _{a}\left( \mathsf{w}^{\text{SOV}}\right) ^{-1}\Theta
^{-1}\mathsf{w}^{\text{SOV}} &=&\left( \eta _{A}\prod_{a=1}^{\SRN-1}\eta
_{a}\right) \text{\textsf{T}}_{\SRN}^{-}, \\
\prod_{a=1}^{\SRN}\xi _{a}\left( \mathsf{w}^{\text{SOV}}\right) ^{-1}\Theta\mathsf{w}^{\text{SOV}} &=&\left( \eta _{D}\prod_{a=1}^{\SRN-1}\eta
_{a}\right) \text{\textsf{T}}_{\SRN}^{+},
\end{eqnarray}
\end{itemize}

Note that taking the average value of the last two formulae we get for $\SRN$ odd:
\begin{equation}
\prod_{a=1}^{\SRN}\frac{X_{a}}{Z_{a}}=\prod_{a=1}^{\SRN}V_{a}^{(-1)^{1+a}},
\end{equation}%
while for $\SRN$ even:
\begin{equation}
Z_{A}=\langle \Theta \rangle
^{-1}\prod_{a=1}^{\SRN-1}Z_{a}^{-1}\prod_{a=1}^{\SRN}X_{a},\text{ \ \ \ \ }%
Z_{D}=Z_{A}\langle \Theta \rangle ^{2},
\end{equation}%
where $\langle \Theta \rangle $ is the average value of the charge $\Theta$.

\section{Comparison with the Fateev-Zamolodchikov model}\label{FZ}

\setcounter{equation}{0}

In this appendix we present an explicit comparison between the SG model, studied in this paper, and the Fateev-Zamolodchikov lattice model with Z$_{p}$ symmetry \cite{FZ}. The Lax operator which describes the FZ model has the following expression in terms of the Lax operator of the SG model:
\begin{equation}\label{Lax-FZ-SG}
L_{n}^{FZ}(\lambda )=L_{n}^{SG}(\lambda )\sigma _{1}.
\end{equation}
In the case $\SRN(=2\SRM)$ even we can construct a map which
transforms the transfer matrix of the SG model into the one of the FZ model. Let us introduce the unitary operators:
\begin{equation}
\Omega _{n}\su_{n}\Omega _{n}=\su_{n}^{-1},\text{ \ \ }\Omega_{n}\sv_{n}\Omega_{n}=\sv_{n}^{-1},
\end{equation}
which in the momentum space play the role of parity operators. Then the unitary operator:
\begin{equation}
\pi _{FZ}\equiv \prod\limits_{n=1}^{\SRM}\Omega _{2n}
\end{equation}
has the following action on the Lax operators:
\begin{equation}
\pi _{FZ}L_{2n-a}^{SG}(\lambda )\pi _{FZ}=\left( \sigma _{1}\right)
^{1-a}L_{2n-a}^{SG}((-1)^{(1-a)}\lambda )\left( \sigma _{1}\right) ^{1-a},\ \  a=0,1.
\end{equation}
so that we get: 
\begin{equation}
\SM^{FZ}(\lambda )=\sigma _{1}\pi _{FZ}\SM^{SG}(\lambda )\pi _{FZ}\sigma _{1}\ \ \longrightarrow\ \ \left\{ 
\begin{array}{c}
\ST^{FZ}(\lambda )=\pi _{FZ}\ST^{SG}(\lambda )\pi _{FZ},
\\ 
\SQ^{FZ}(\lambda )=\pi _{FZ}\SQ^{SG}(\lambda )\pi _{FZ},
\end{array}%
\right.
\end{equation}
after the flipping $\xi _{2n-a}\ \rightarrow \ \left( -1\right) ^{1-a}\xi_{2n-a}$ of the inhomogeneities.

In the case $\SRN$ odd the situation is different; the transfer matrices in the two model have different spectrum. We use the next two subsections to present an explicit comparison of their spectrum in the special case of $q^{3}=1$ and $\SRN=1$.

\subsection{$\SQ$-spectrum in Sine-Gordon model for $q^{3}=1$ and $\SRN=1$}

In this case in the $\bz$-representation the operator \textit{Q}$_{\text{SG}%
}(\lambda )$ is a $3\times 3$ matrix\footnote{Note that to make more simple the comparison with the $\SQ$-operator of the
FZ model, here we have considered for {\it Q}$_{\text{SG}}(\lambda )$ the operator
$\SY(\la)$ defined in \rf{Ydef} just with a different
normalization.}: 
\begin{equation}
\text{\textit{Q}}_{\text{SG}}(\lambda)\equiv ||\langle z=q^{2(i-1)}|\text{%
\textit{Q}}_{\text{SG}}(\lambda )|z^{\prime }=q^{2(j-1)}\rangle \equiv
W_{\lambda_{+} }(q^{2(i+j-2)})\text{$\overline{W}$}_{\lambda_{-}
}(q^{2(i-j)})||_{i,j\in \{1,2,3\}}
\end{equation}
and $\lambda _{\pm }\equiv \epsilon \lambda \kappa ^{\pm }$. Now, we observe
that 
\begin{align}
W_{\lambda }(1) \equiv 1,\text{ }W_{\lambda }(q^{2})=W_{\lambda }(q^{4})=%
\frac{1+\lambda q}{\lambda +q}\equiv \text{\textsc{w}}_{\lambda },\,\,\,\,\,\,\,\,\,\, \\
\overline{W}_{\lambda }(1) \equiv 1,\text{ \ $\overline{W}$}%
_{\lambda }(q^{2})=\text{$\overline{W}$}_{\lambda }(q^{4})=\frac{\lambda
-1}{\lambda q - q^{-1}}\equiv \overline{\text{\textsc{w}}}_{\lambda },
\end{align}
so that in the $\bz$-representation: 
\begin{equation}
\text{\textit{Q}}_{\text{SG}}(\lambda )\equiv \left( 
\begin{array}{ccc}
1 & \text{\textsc{w}}_{\lambda _{+}}\overline{\text{\textsc{w}}}_{\lambda
_{-}} & \text{\textsc{w}}_{\lambda _{+}}\overline{\text{\textsc{w}}}%
_{\lambda _{-}} \\ 
\text{\textsc{w}}_{\lambda _{+}}\overline{\text{\textsc{w}}}_{\lambda _{-}}
& \text{\textsc{w}}_{\lambda _{+}} & \overline{\text{\textsc{w}}}_{\lambda
_{-}} \\ 
\text{\textsc{w}}_{\lambda _{+}}\overline{\text{\textsc{w}}}_{\lambda _{-}}
& \overline{\text{\textsc{w}}}_{\lambda _{-}} & \text{\textsc{w}}_{\lambda
_{+}}
\end{array}
\right) .
\end{equation}
Then the eigenvalues of {\it Q}$_{\text{SG}}(\lambda )$
read: 
\begin{equation}
q_{1}^{(SG)}(\lambda )=(\text{\textsc{w}}_{\lambda _{+}}-\overline{\text{%
\textsc{w}}}_{\lambda _{-}}),\text{ \ \ }q_{\pm }^{(SG)}(\lambda )=\frac{1}{2%
}\left( 1+\text{\textsc{w}}_{\lambda _{+}}+\overline{\text{\textsc{w}}}%
_{\lambda _{-}}\pm \Delta _{\lambda }\right) ,
\end{equation}
with $\Delta _{\lambda }\equiv \left( (\text{\textsc{w}}_{\lambda
_{+}}-1)^{2}+2(\text{\textsc{w}}_{\lambda _{+}}-1)\overline{\text{\textsc{w}}%
}_{\lambda _{-}}+(1+8\text{\textsc{w}}{}_{\lambda _{+}}^{2})\overline{\text{%
\textsc{w}}}{}_{\lambda _{-}}^{2}\right) ^{1/2}$ and clearly {\it Q}$_{\text{SG}}(\lambda )$ has simple spectrum for all the values of the local parameter $\kappa \in \mathbb{C}$.

\subsection{$\SQ$-spectrum in  Fateev-Zamolodchikov model for $q^{3}=1$ and $\SRN=1$}

In this case in the z-representation the operator {\it Q}$_{\text{FZ}}(\lambda )$
is a $3\times 3$ matrix\footnote{Here we have rewritten in our notation the (5.12) of \cite{BS} for $k=0$.}: 
\begin{equation}
\text{\textit{Q}}_{\text{FZ}}(\lambda )\equiv ||\langle z=q^{2(i-1)}|\text{\textit{Q}}_{\text{%
FZ}}(\lambda )|z^{\prime }=q^{2(j-1)}\rangle \equiv W_{\lambda }(q^{2(i-j)})%
\text{$\overline{W}$}_{\lambda }(q^{2(i-j)})||_{i,j\in \{1,2,3\}}\text{,}
\end{equation}
explicitly: 
\begin{equation}
\text{\textit{Q}}_{\text{FZ}}(\lambda )\equiv \left( 
\begin{array}{ccc}
1 & \text{\textsc{w}}_{\lambda }\overline{\text{\textsc{w}}}_{\lambda } & 
\text{\textsc{w}}_{\lambda }\overline{\text{\textsc{w}}}_{\lambda } \\ 
\text{\textsc{w}}_{\lambda }\overline{\text{\textsc{w}}}_{\lambda } & 1 & 
\text{\textsc{w}}_{\lambda }\overline{\text{\textsc{w}}}_{\lambda } \\ 
\text{\textsc{w}}_{\lambda }\overline{\text{\textsc{w}}}_{\lambda } & \text{%
\textsc{w}}_{\lambda }\overline{\text{\textsc{w}}}_{\lambda } & 1
\end{array}
\right) .
\end{equation}
It is then clear that {\it Q}$_{\text{FZ}}(\lambda )$ has degenerate spectrum with eigenvalues: 
\begin{equation}
q_{1}^{(FZ)}(\lambda )=1+2\text{\textsc{w}}_{\lambda }\overline{\text{%
\textsc{w}}}_{\lambda },\text{ \ \ }q_{\pm }^{(FZ)}(\lambda )=1-\text{%
\textsc{w}}_{\lambda }\overline{\text{\textsc{w}}}_{\lambda }.
\end{equation}

\end{document}